\documentclass{article}

\usepackage{PRIMEarxiv}

\usepackage[utf8]{inputenc} 
\usepackage[T1]{fontenc}    
\usepackage{hyperref}       
\usepackage{url}            
\usepackage{booktabs}       
\usepackage{amsmath}
\usepackage{amsthm}
\usepackage{amssymb}
\usepackage{amsfonts}
\usepackage{mathtools}
\usepackage{amsfonts}       
\usepackage{nicefrac}       
\usepackage{microtype}      
\usepackage{fancyhdr}       
\usepackage{graphicx}       
\usepackage{times}
\usepackage{subfig}
\usepackage{xcolor}
\usepackage{array}
\usepackage{cite}
\usepackage{multirow}
\usepackage{adjustbox}
\usepackage{fontawesome}
\usepackage[font=footnotesize,labelfont=bf]{caption}
\usepackage{tikz-network}
\usepackage{enumitem}

\DeclareMathOperator{\diag}{diag}

\DeclareMathOperator{\rank}{rank}

\newtheorem{theorem}{Theorem}
\newtheorem{lemma}{Lemma}

\newtheorem{corollary}{Corollary}

\newtheorem{assumption}{Assumption}

\title{
Learning-Oriented Fault-Tolerant Control for UAVs with Time-Varying Fault Estimation
}

\author{
  Moh Kamalul Wafi \\
  Department of Electrical and Computer Engineering \\
  Northeastern University \\
  Boston, MA, USA\\
  \texttt{\{wafi.m\}@northeastern.edu} \\
}

\begin{document}
\maketitle

\begin{abstract}
The increasing deployment of unmanned aerial vehicles (UAVs) in safety-critical applications has created a growing demand for control strategies that remain reliable in the presence of actuator and sensor faults. This paper presents a fault-tolerant control framework that combines constructive mathematical modeling, proportional--integral (PI) fault estimation, and learning-oriented control. A nonlinear quadcopter model is first developed using both the Newton--Euler and Euler--Lagrange formulations, from which a discrete-time linear model is derived for observer and controller design. A PI observer is then proposed to estimate actuator and sensor faults, with exponential estimation-error convergence established for constant faults in the absence of disturbances and ultimate boundedness established for slowly time-varying faults under bounded disturbances. The estimated fault signals are incorporated into an active fault-tolerant control architecture together with a Deep Deterministic Policy Gradient (DDPG) controller, which is trained using only the yaw channel and subsequently applied to improve closed-loop tracking performance. Numerical simulations illustrate the observer's estimation performance under constant and slowly time-varying faults, including bounded disturbances, and the tracking performance of the integrated controller under multiple actuator and sensor fault scenarios. The proposed framework combines model-based fault estimation with learning-based control to improve fault accommodation and tracking performance in UAV systems.
\end{abstract}
\allowdisplaybreaks

\section{Introduction}

Unmanned aerial systems (UAS) have become an indispensable technology in both civilian and military applications, including surveillance, inspection, environmental monitoring, search-and-rescue missions, and autonomous transportation \cite{R1,R2}. Despite their growing deployment, UAS often operate in uncertain environments where adverse weather conditions, actuator degradation, sensor malfunctions, and external disturbances may significantly deteriorate the control performance and even compromise flight safety. Consequently, the development of resilient fault-tolerant control (FTC) strategies has become an increasingly important research topic for improving the reliability, robustness, and autonomy of unmanned aerial vehicles (UAVs) \cite{R3}.

A fundamental component of any active FTC architecture is the ability to accurately estimate system faults. Fault detection and diagnosis have therefore attracted considerable attention in recent years, particularly for actuator and sensor faults that may lead to performance degradation, instability, or complete mission failure if left uncompensated \cite{R33,R34,R42}. Reliable fault estimation becomes increasingly challenging due to the highly coupled nonlinear dynamics of UAVs together with their multi-input multi-output control architecture, where numerous sensors and actuators must operate simultaneously to achieve the desired flight objectives.

Accurate mathematical modeling is equally important for the successful implementation of fault estimation and control algorithms \cite{WAFI-JRC,Wafi-Hydraulic-arXiv}. Well-identified UAV dynamics provide the foundation for designing controllers ranging from classical linear techniques to advanced learning-based methods \cite{R5}. Beyond the vehicle dynamics themselves, recent research has also emphasized reliable communication and information exchange, including resilient UAV communication networks \cite{R7,Wafi-ASCC26-DNC}, optimal navigation strategies \cite{R6}, UAV-to-ground communication architectures \cite{Wafi-AIP,Wafi-AIP-arXiv,R48}, and secure communication frameworks for both single- and multi-UAV systems \cite{R8,R9}. These developments further motivate the need for robust control architectures capable of maintaining satisfactory performance in the presence of component failures and environmental uncertainties.

Among existing fault estimation methods, proportional--integral (PI) observers have received considerable attention because of their simple structure, low computational complexity, and capability to simultaneously estimate system states and unknown faults \cite{R44,R45,R46}. Once the faults are accurately reconstructed, they can be compensated within an active FTC framework, thereby recovering the nominal closed-loop behavior despite actuator and sensor degradations \cite{R23,R29}. In parallel, deep reinforcement learning (DRL) has emerged as a promising tool for learning complex control policies directly from system interactions without requiring an accurate analytical model \cite{R47,Wafi-SIAM,Wafi-SIAM-arXiv}. Recent studies have demonstrated the effectiveness of learning-based FTC strategies by integrating DRL with conventional control architectures \cite{R49,R50,R51}. Motivated by these advances, this paper combines a PI observer with a Deep Deterministic Policy Gradient (DDPG) controller. In particular, unlike conventional approaches that require learning for every controlled variable, the proposed framework trains the DDPG agent using only a single output channel while exploiting the learned policy to improve the control performance of multiple system states.

The remainder of this paper is organized as follows. Section~II develops the nonlinear dynamic model of the UAV. Section~III presents the fault modeling and problem formulation. Section~IV develops the PI observer for actuator and sensor fault estimation. Section~V introduces the proposed fault-tolerant control framework. Section~VI describes the learning-oriented DDPG controller. Numerical simulation results are presented in Section~VII, followed by concluding remarks and future research directions in Section~VIII.

\section{Dynamics of Quadcopter}
\label{sec:2}

\subsection{Mathematical Model of Quadcopter}

The dynamics of the quadcopter are illustrated in Fig.~\ref{F1}, where the vehicle is described with respect to two reference frames: the earth-inertial frame ($E$) and the body-fixed frame ($B$). The figure also depicts the forces and torques generated by the four rotors, indexed by $i=1,2,3,4$. The inertial frame is fixed to the earth, with gravity acting along the negative $z$-axis, whereas the body frame is rigidly attached to the quadcopter. Each arm is equipped with a propeller driven by a motor that produces a thrust through its angular velocity $\omega_i$ and the corresponding motor torque $\tau_{m_i}$. Rotors $1$ and $3$ rotate counterclockwise, while rotors $2$ and $4$ rotate clockwise to balance the reaction torques.

The position of the quadcopter's center of mass in the inertial frame is denoted by $\zeta=[x,y,z]^\top$, whereas its orientation is described by the Euler angles $\eta=[\phi,\theta,\psi]^\top$, representing the roll $(\phi)$, pitch $(\theta)$, and yaw $(\psi)$ angles about the $x$-, $y$-, and $z$-axes, respectively. Accordingly,
\begin{figure}[t!]
    \centering
    \includegraphics[width=.85\linewidth]{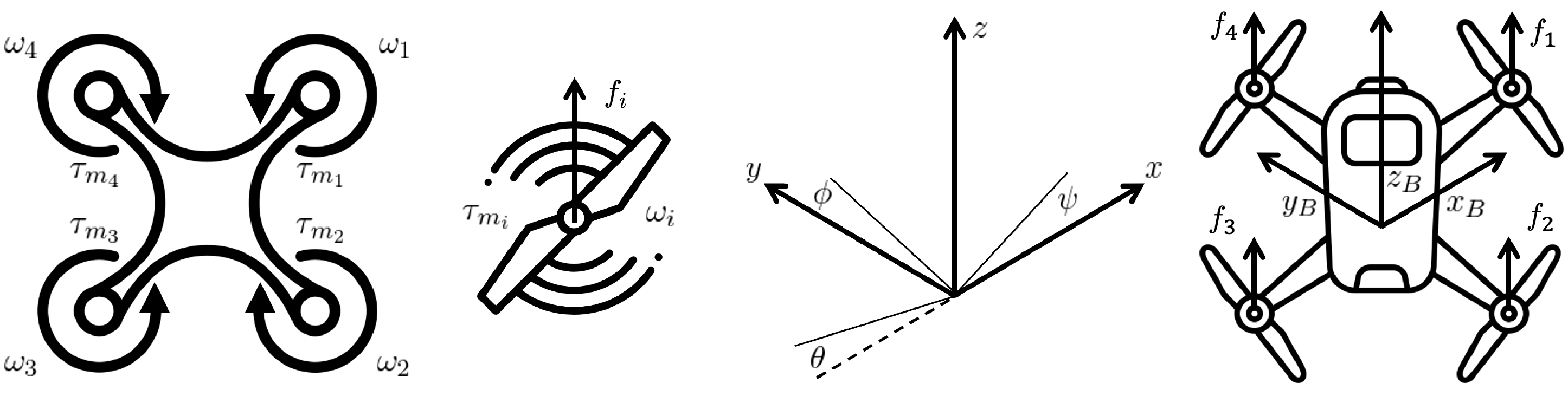}
    \caption{The figure shows the structural frames of a quadcopter as viewed from both the inertial and body perspectives.}
    \label{F1}
\end{figure}
\begin{subequations}
\begin{align}
    \zeta &= \left[\begin{array}{c}
        x \\
        y \\
        z
    \end{array}\right], \quad \eta = \left[\begin{array}{c}
        \phi \\
        \theta \\
        \psi
    \end{array}\right], \quad \mu_p = \left[\begin{array}{c}
        \zeta \\
        \eta
    \end{array}\right], \label{eq:position_states}
\end{align}
and
\begin{align}
    V_B &= \left[\begin{array}{c}
        u \\
        v \\
        w
    \end{array}\right], \quad \vartheta = \left[\begin{array}{c}
        p \\
        q \\
        r
    \end{array}\right], \quad \mu_v = \left[\begin{array}{c}
        V_B \\
        \vartheta
    \end{array}\right]. \label{eq:velocity_states}
\end{align}
\end{subequations}
The vectors $\mu_p$ and $\mu_v$ collect the translational and rotational position and velocity variables, respectively. Specifically, $V_B=[u,v,w]^\top$ denotes the linear velocity expressed in the body-fixed frame, whereas $\vartheta=[p,q,r]^\top$ represents the corresponding angular velocity. The body-fixed frame is attached to the quadcopter's center of mass, and its orientation with respect to the inertial frame is described by the rotation matrix $R$.

For notational convenience, we use the shorthand $\mathbf{c}_x=\cos(x)$, $\mathbf{s}_x=\sin(x)$, and $\mathbf{t}_x=\tan(x)$. The elementary rotation matrices about the $x$-, $y$-, and $z$-axes are denoted by $R_x(\phi)$, $R_y(\theta)$, and $R_z(\psi)$, respectively, and are given by
\begin{subequations}
\begin{align}
    \left[\begin{array}{c|c|c}
        R_x(\phi) & R_y(\theta) & R_z(\psi)
    \end{array}\right]
    =
    \left[\begin{array}{ccc|ccc|ccc}
        1 & 0 & 0 & \textbf{c}_\theta & 0 & \textbf{s}_\theta & \textbf{c}_\psi & -\textbf{s}_\psi & 0\\
        0 & \textbf{c}_\phi & -\textbf{s}_\phi & 0 & 1 & 0 & \textbf{s}_\psi & \textbf{c}_\psi & 0\\
        0 & \textbf{s}_\phi & \textbf{c}_\phi & -\textbf{s}_\theta & 0 & \textbf{c}_\theta & 0 & 0 & 1
    \end{array}\right].
    \label{eq:elementary_rotations}
\end{align}

The rotation from the body-fixed frame to the inertial frame is obtained using the standard $ZYX$ Euler-angle convention,
\begin{align}
    R_{B\to E}
    \coloneqq
    R_z(\psi)R_y(\theta)R_x(\phi),
    \label{eq:body_to_inertial}
\end{align}
which yields the complete rotation matrix
\begin{gather}
    R =
    \left[\begin{array}{ccc}
        \textbf{c}_\psi \textbf{c}_\theta &
        \textbf{c}_\psi \textbf{s}_\theta \textbf{s}_\phi - \textbf{s}_\psi \textbf{c}_\phi &
        \textbf{c}_\psi \textbf{s}_\theta \textbf{c}_\phi + \textbf{s}_\psi \textbf{s}_\phi\\
        \textbf{s}_\psi \textbf{c}_\theta &
        \textbf{s}_\psi \textbf{s}_\theta \textbf{s}_\phi + \textbf{c}_\psi \textbf{c}_\phi &
        \textbf{s}_\psi \textbf{s}_\theta \textbf{c}_\phi - \textbf{c}_\psi \textbf{s}_\phi\\
        -\textbf{s}_\theta &
        \textbf{c}_\theta \textbf{s}_\phi &
        \textbf{c}_\theta \textbf{c}_\phi
    \end{array}\right].
    \label{eq:rotation_matrix}
\end{gather}
\end{subequations}
Since $R$ is the product of orthogonal rotation matrices, it is itself orthogonal. Consequently, $R^{-1}=R^\top.$

The four propeller--motor assemblies are assumed to be identical. Consequently, it is sufficient to derive the thrust and torque generated by a single rotor before extending the results to the complete quadcopter.
Each rotor generates an individual thrust force $f_i$, and the total thrust $T$ is obtained by summing the contributions from all four rotors,
\begin{align}
    f_i = \left(C_d\rho Ar^2\right)\omega_i^2,
    \label{eq:thrust}
\end{align}
where $C_d$ is the thrust coefficient, $\rho$ is the air density, $A$ is the propeller cross-sectional area, and $r$ is the rotor radius. The constant $k\coloneqq C_d\rho Ar^2$
is introduced as the lift coefficient, allowing the total thrust to be written compactly as
\[
    T = \sum_{i=1}^4 f_i =  k\sum_{i=1}^4\omega_i^2.
\]
The corresponding motor torque is modeled as
\begin{align}
    \tau_{m_i} = b\omega_i^2 + I_m\dot{\omega}_i,
    \label{eq:motor_torque}
\end{align}
where $b$ denotes the aerodynamic drag coefficient and $I_m$ is the rotor inertia. Since the rotor dynamics are considerably faster than the rigid-body dynamics, the angular acceleration is neglected, i.e.,
\[
    \dot{\omega}_i\approx0.
\]
When all four rotors rotate at the same angular speed, the resulting thrust acts solely along the body-fixed $z$-axis, enabling the quadcopter to ascend, descend, or hover according to its magnitude relative to gravity. By contrast, differences in the rotor thrusts generate the control moments associated with the roll $(\phi)$, pitch $(\theta)$, and yaw $(\psi)$ motions. Hence, the net thrust vector $T_B$ and the control torque vector $\tau_B$ take the form
\begin{align}
    T_B = \left[\begin{array}{c}
        0 \\
        0 \\
        T
    \end{array}\right], \quad \tau_B = \left[\begin{array}{c}
        \tau_\phi \\
        \tau_\theta \\
        \tau_\psi
    \end{array}\right] = \left[\begin{array}{c}
        L(f_4 - f_2) \\
        L(f_3 - f_1) \\
        \tau_{m_1}-\tau_{m_2}+\tau_{m_3}-\tau_{m_4}
    \end{array}\right],
    \label{eq:thrust_and_torque}
\end{align}
where $L$ denotes the distance from each rotor to the quadcopter's center of mass. The roll torque $\tau_\phi$ is generated by the thrust difference between rotors $4$ and $2$, written as $L(f_4-f_2)$, whereas the pitch torque $\tau_\theta$ is generated by the thrust difference between rotors $3$ and $1$, given by $L(f_3-f_1)$. These rotor pairs are located on opposite sides of the vehicle\footnote{Specifically, rotors $(2,4)$ form one pair and rotors $(1,3)$ form the other pair, as illustrated in Fig.~\ref{F1}.}. The yaw torque $\tau_\psi$ is produced by the imbalance of the reaction torques generated by the counter-rotating propellers. Throughout this work, the quadcopter is assumed to be geometrically symmetric, implying
\[
    I_{xx}=I_{yy},
\]
with inertia matrix
\[
    I=\diag\{I_{xx},I_{yy},I_{zz}\}.
\]
Under the rigid-body assumption, the translational dynamics are obtained by applying Newton's second law in the body-fixed frame. The inertial force associated with the mass acceleration $m\dot{V}_B$ and the centrifugal force $\vartheta\times(mV_B)$ are balanced by the gravitational force $R^\top\varrho$ and the net thrust $T_B$, yielding
\begin{align}
    m\dot{V}_B + \vartheta \times \left(mV_B\right) = R^\top\varrho + T_B,
    \label{eq:body_translational_dynamics}
\end{align}
Using the kinematic relation $\dot{\zeta}=RV_B$, the translational dynamics can be expressed in the inertial frame. The rotational term $\vartheta\times(mV_B)$ combines with the derivative of the rotation matrix to recover the inertial-frame acceleration,
\begin{align*}
    m\ddot{\zeta}=m\dot{R}V_B+mR\dot{V}_B.
\end{align*}
Using the identity $\dot{R}V_B=R(\vartheta\times V_B),$ it follows that
\[
    m\ddot{\zeta}=R\left[m\dot{V}_B+\vartheta\times(mV_B)\right].
\]
Substituting \eqref{eq:body_translational_dynamics} together with the orthogonality property $RR^\top=I$ gives
\[
    m\ddot{\zeta}=R(R^\top\varrho+T_B)=\varrho+RT_B,
\]
or equivalently,
\begin{align}\begin{aligned}
    \left[\begin{array}{c}
        \ddot{x} \\
        \ddot{y} \\
        \ddot{z}
    \end{array}\right] &= -g\left[\begin{array}{c}
        0 \\
        0 \\
        1
    \end{array}\right] + \frac{T}{m}\left[\begin{array}{c}
        \textbf{c}_\psi \textbf{s}_\theta \textbf{c}_\phi + \textbf{s}_\psi \textbf{s}_\phi \\
        \textbf{s}_\psi \textbf{s}_\theta \textbf{c}_\phi - \textbf{c}_\psi \textbf{s}_\phi \\
        \textbf{c}_\theta \textbf{c}_\phi
    \end{array}\right].
\end{aligned}
\label{eq:inertial_translational_dynamics}
\end{align}
The translational dynamics describe the evolution of the quadcopter position $\zeta$, whereas the attitude dynamics governing the Euler angles $\eta$ are obtained from the rotational motion of the rigid body.

Similarly, the rotational dynamics are obtained by applying Euler's equation to the rigid body. The angular acceleration term $I\dot{\vartheta}$, together with the centripetal term $\vartheta\times(I\vartheta)$ and the gyroscopic torque $\Gamma$, balances the control torque $\tau$, yielding
\begin{align}
    I\dot{\vartheta} + \vartheta \times (I\vartheta) + \Gamma = \tau.
    \label{eq:body_rotational_dynamics}
\end{align}
Expanding \eqref{eq:body_rotational_dynamics} gives the angular acceleration in terms of the body angular velocity $\vartheta=[p,q,r]^\top$, the control torque $\tau=[\tau_\phi,\tau_\theta,\tau_\psi]^\top$, and the gyroscopic effect,
\begin{align}
    \dot{\vartheta} = I^{-1}\left(-\vartheta\times\begin{bmatrix}
        I_{xx}p \\
        I_{yy}q \\
        I_{zz}r
    \end{bmatrix} - I_r \vartheta\times\begin{bmatrix}
        0 \\
        0 \\
        1
    \end{bmatrix}\Omega_\Gamma + \tau\right).
    \label{eq:angular_acceleration}
\end{align}
Using the symmetry assumption $I_{xx}=I_{yy}$ together with the rotor angular momentum
\[
    \Omega_\Gamma=\omega_1-\omega_2+\omega_3-\omega_4,
\]
the body angular accelerations about the roll $(\phi)$, pitch $(\theta)$, and yaw $(\psi)$ axes become
\begin{align}\begin{aligned}
    \dot{p} &= \frac{(I_{yy} - I_{zz})qr}{I_{xx}} - I_r \frac{q}{I_{xx}}\Omega_\Gamma + \frac{\tau_\phi}{I_{xx}},\\
    \dot{q} &= \frac{(I_{zz} - I_{xx})pr}{I_{yy}} + I_r \frac{p}{I_{yy}}\Omega_\Gamma + \frac{\tau_\theta}{I_{yy}},\\
    \dot{r} &= \frac{(I_{xx} - I_{yy})pq}{I_{zz}} + \frac{\tau_\psi}{I_{zz}}.
\end{aligned}
\label{eq:body_rate_dynamics}
\end{align}

The body angular velocity $\vartheta=[p,q,r]^\top$ is related to the Euler-angle rate $\dot{\eta}=[\dot{\phi},\dot{\theta},\dot{\psi}]^\top$ through the kinematic transformation matrices $W_\eta^{-1}$ and $W_\eta$.\footnote{The Euler-angle rates $\dot{\eta}$ are
coordinate derivatives and are generally different from the body angular velocity $\vartheta$.
The kinematic transformation relates these quantities.} These transformations relate the body angular velocity
to the Euler-angle rates according to
\begin{align}\begin{aligned}
    \dot{\eta} = W_\eta^{-1}\vartheta, &\begin{bmatrix}
        \dot{\phi} \\
        \dot{\theta} \\
        \dot{\psi}
    \end{bmatrix} &= \begin{bmatrix}
        1 & \textbf{s}_\phi \textbf{t}_\theta & \textbf{c}_\phi \textbf{t}_\theta \\
        0 & \textbf{c}_\phi & -\textbf{s}_\phi \\
        0 & \textbf{s}_\phi/\textbf{c}_\theta & \textbf{c}_\phi/\textbf{c}_\theta
    \end{bmatrix}\begin{bmatrix}
        p \\
        q \\
        r
    \end{bmatrix},\qquad
    \vartheta = W_\eta\dot{\eta}, &\begin{bmatrix}
        p \\
        q \\
        r
    \end{bmatrix} &= \begin{bmatrix}
        1 & 0 & -\textbf{s}_\theta \\
        0 & \textbf{c}_\phi & \textbf{s}_\phi \textbf{c}_\theta \\
        0 & -\textbf{s}_\phi & \textbf{c}_\phi \textbf{c}_\theta
    \end{bmatrix}\begin{bmatrix}
        \dot{\phi} \\
        \dot{\theta} \\
        \dot{\psi}
    \end{bmatrix}.
\end{aligned}
\label{eq:euler_kinematics}
\end{align}

Differentiating the kinematic relation $\dot{\eta}=W_\eta^{-1}\vartheta$ with respect to time gives
\begin{align}
    \ddot{\eta}
    =
    \frac{d}{dt}\left(W_\eta^{-1}\vartheta\right)
    =
    \frac{d}{dt}\left(W_\eta^{-1}\right)\vartheta
    +
    W_\eta^{-1}\dot{\vartheta},
    \label{eq:euler_acceleration}
\end{align}
where the time derivative of the transformation matrix $W_\eta^{-1}$ is given by
\begin{align*}
    \dot{W}_\eta^{-1}
    =
    \left[\begin{array}{ccc}
        0 &
        \dot{\phi}\textbf{c}_\phi \textbf{t}_\theta + \dot{\theta}\textbf{s}_\phi/\textbf{c}_\theta^2 &
        -\dot{\phi}\textbf{s}_\phi \textbf{t}_\theta + \dot{\theta}\textbf{c}_\phi/\textbf{c}_\theta^2 \\
        0 &
        -\dot{\phi}\textbf{s}_\phi &
        -\dot{\phi}\textbf{c}_\phi \\
        0 &
        \dot{\phi}\textbf{c}_\phi/\textbf{c}_\theta + \dot{\theta}\textbf{s}_\phi \textbf{t}_\theta/\textbf{c}_\theta &
        -\dot{\phi}\textbf{s}_\phi/\textbf{c}_\theta + \dot{\theta}\textbf{c}_\phi \textbf{t}_\theta/\textbf{c}_\theta
    \end{array}\right].
\end{align*}

\subsection{Alternative Lagrangian Derivation of the Equations of Motion}

For completeness, the equations of motion can also be derived using the Euler--Lagrange formulation. Unlike the Newton--Euler approach presented previously, the Lagrangian framework is based on the kinetic and potential energies of the system. 
Let $\mu=\mu_p=[\zeta^\top,\eta^\top]^\top$ denote the generalized coordinates.
The Lagrangian function is defined as
\begin{align}
    \mathcal{L}(\mu,\dot{\mu})
    =
    E_T+E_R-E_P
    =
    \frac{m}{2}\dot{\zeta}^{\top}\dot{\zeta}
    +
    \frac{1}{2}\vartheta^{\top}I\vartheta
    -
    mgz,
    \label{eq:lagrangian}
\end{align}
where the first, second, and third terms represent the translational kinetic energy, rotational kinetic energy, and gravitational potential energy, respectively.

The Euler--Lagrange equations in the presence of
generalized forces $f$ and torques $\tau_\eta$ are given by
\begin{align}
    \begin{bmatrix}
        f\\
        \tau_\eta
    \end{bmatrix}
    =
    \frac{d}{dt}
    \left(
    \frac{\partial\mathcal{L}}
         {\partial\dot{\mu}}
    \right)
    -
    \frac{\partial\mathcal{L}}
         {\partial\mu}.
    \label{eq:euler_lagrange}
\end{align}

Since the translational and rotational coordinates are independent, the corresponding equations can be derived separately. For the translational motion, the generalized force is the net thrust generated by the four rotors,
\begin{align}
    f
    =
    RT_B
    =
    m\ddot{\zeta}
    +
    mg
    \begin{bmatrix}
        0&0&1
    \end{bmatrix}^{\!\top},
    \label{eq:lagrangian_translation}
\end{align}
which is equivalent to the translational dynamics obtained in \eqref{eq:inertial_translational_dynamics}.

To derive the rotational dynamics, the body angular velocity $\vartheta$ is related to the Euler-angle rate $\dot{\eta}$ through the transformation matrix $W_\eta$. Consequently, the generalized inertia matrix with respect to the Euler coordinates is
\begin{align}
    J(\eta)
    \coloneqq
    W_\eta^{\top}IW_\eta
    =
    \begin{bmatrix}
        I_{xx}
        &
        0
        &
        -I_{xx}\mathbf{s}_{\theta}
        \\
        0
        &
        I_{yy}\mathbf{c}_{\phi}^{2}
        +
        I_{zz}\mathbf{s}_{\phi}^{2}
        &
        \xi_{1}
        \\
        -I_{xx}\mathbf{s}_{\theta}
        &
        \xi_{1}
        &
        \xi_{2}
    \end{bmatrix},
    \label{eq:generalized_inertia}
\end{align}
where $\xi_1=(I_{yy}-I_{zz})\mathbf{c}_{\phi}\mathbf{s}_{\phi}\mathbf{c}_{\theta}$ and
$\xi_2=I_{xx}\mathbf{s}_{\theta}^{2}+I_{yy}\mathbf{s}_{\phi}^{2}\mathbf{c}_{\theta}^{2}
+I_{zz}\mathbf{c}_{\phi}^{2}\mathbf{c}_{\theta}^{2}.$

The rotational kinetic energy can therefore be expressed in the Euler coordinates as
\begin{align}
    E_R
    =
    \frac12
    \vartheta^\top I\vartheta
    =
    \frac12
    \dot{\eta}^{\top}
    J(\eta)
    \dot{\eta},
    \label{eq:rotational_energy}
\end{align}
where we have used the relation $\vartheta=W_\eta\dot{\eta}$.

Applying the Euler--Lagrange equation to the rotational coordinates yields
\begin{subequations}
\begin{align}
    \tau_\eta
    =
    W_\eta^\top(\tau_B-\Gamma)
    =
    J(\eta)\ddot{\eta}
    +
    \dot{J}(\eta)\dot{\eta}
    -
    \frac12
    \frac{\partial}
         {\partial\eta}
    \left(
    \dot{\eta}^{\top}
    J(\eta)
    \dot{\eta}
    \right),
    \label{eq:lagrangian_rotational_dynamics}
\end{align}
where the last two terms are commonly grouped as
$C(\eta,\dot{\eta})\dot{\eta}$, representing the
Coriolis and centrifugal effects.
Consequently,
\begin{align*}
    W_\eta^\top(\tau_B-\Gamma)
    =
    J(\eta)\ddot{\eta}
    +
    C(\eta,\dot{\eta})\dot{\eta},
\end{align*}
and the angular acceleration is obtained as
\begin{align}
    \ddot{\eta}
    =
    J^{-1}(\eta)
    \left[
        W_\eta^\top(\tau_B-\Gamma)
        -
        C(\eta,\dot{\eta})\dot{\eta}
    \right],
    \label{eq:lagrangian_angular_acceleration}
\end{align}
\end{subequations}
which is equivalent to the rotational dynamics derived previously.

To account for aerodynamic effects encountered during flight, the translational dynamics are augmented with a drag force. The body-frame and inertial-frame equations become
\begin{subequations}
\begin{gather}
    m\dot{V}_B+\vartheta\times(mV_B)
    =
    R^\top\varrho
    +
    T_B
    -
    R^\top T_D,
    \label{eq:body_translational_dynamics_drag}
    \\
    m\ddot{\zeta}
    =
    \varrho
    +
    RT_B
    -
    T_D,
    \label{eq:inertial_translational_dynamics_drag}
\end{gather}
\end{subequations}
where the drag force is modeled as
\begin{align}
    T_D
    =
    A\dot{\zeta}
    =
    \begin{bmatrix}
    A_x&0&0\\
    0&A_y&0\\
    0&0&A_z
    \end{bmatrix}
    \begin{bmatrix}
    \dot{x}\\
    \dot{y}\\
    \dot{z}
    \end{bmatrix},
    \label{eq:drag_force}
\end{align}
with $A_x$, $A_y$, and $A_z$ denoting the drag coefficients along the inertial $x$-, $y$-, and $z$-axes, respectively.

Finally, the individual rotor angular velocities can be recovered from the total thrust $T$ and the control torques $\tau_\phi$, $\tau_\theta$, and $\tau_\psi$ using \eqref{eq:motor_torque} and \eqref{eq:thrust_and_torque},
\begin{align}
    \begin{aligned}
    \omega_1^2
    &=
    \frac{T}{4k}
    -
    \frac{\tau_\theta}{2kL}
    +
    \frac{\tau_\psi}{4b},
    \qquad
    \omega_2^2
    =
    \frac{T}{4k}
    -
    \frac{\tau_\phi}{2kL}
    -
    \frac{\tau_\psi}{4b},
    \\
    \omega_3^2
    &=
    \frac{T}{4k}
    +
    \frac{\tau_\theta}{2kL}
    +
    \frac{\tau_\psi}{4b},
    \qquad
    \omega_4^2
    =
    \frac{T}{4k}
    +
    \frac{\tau_\phi}{2kL}
    -
    \frac{\tau_\psi}{4b}.
    \end{aligned}
    \label{eq:rotor_speeds}
\end{align}

The resulting nonlinear model forms the basis for the subsequent fault estimation and fault-tolerant learning-oriented control framework. In the next subsection, the nonlinear dynamics are linearized and augmented to incorporate actuator and sensor fault states together with external disturbances.

\subsection{Linearized System and Control Inputs}

The complete quadcopter model consists of twelve states, corresponding to the translational and rotational positions together with their associated velocities. Among these, the position and attitude variables
\[
    [x,y,z,\phi,\theta,\psi]
\]
represent the primary quantities to be regulated and measured. As commonly adopted in quadcopter control, the translational motions $(x,y)$ are governed by an outer-loop position controller, whereas the altitude and attitude variables $(z,\phi,\theta,\psi)$ are regulated by an inner-loop controller.

Since this work focuses on attitude stabilization and fault-tolerant control, only the inner-loop dynamics are considered. Accordingly, the state vector is defined as
\[
    x=[z,\phi,\theta,\psi,w,p,q,r]^\top \in\mathbb{R}^{n},
\]
where $w$ denotes the body-frame vertical velocity,
and $p$, $q$, and $r$ are the body angular velocities.
The measured output consists of the altitude and Euler angles,
\[
    y=[z,\phi,\theta,\psi]^\top.
\]

The nonlinear dynamics
\[
    \dot{x}(t)=f(x,t)+g(x,t)u
\]
are linearized about the hover equilibrium
\[
    x^\star=[z^\star,0,\ldots,0]^\top,
\]
with the corresponding equilibrium input $u^\star=[mg,0,0,0]^\top.$
The linearization is expressed in terms of the
deviation variables $\delta x=x-x^\star$ and
$\delta u=u-u^\star$, and subsequently discretized
to obtain the state-space model used for state
estimation and controller design.

The control input is composed of the collective thrust and the three body torques,
\begin{align*}
    u=
    \begin{bmatrix}
    T\\
    \tau_\phi\\
    \tau_\theta\\
    \tau_\psi
    \end{bmatrix},
\end{align*}
where $T$ regulates the altitude, while $\tau_\phi$, $\tau_\theta$, and $\tau_\psi$ control the roll $(\phi)$, pitch $(\theta)$, and yaw $(\psi)$ motions, respectively.

To exploit the advantages of both classical and learning-based control, a hybrid control architecture is adopted. Specifically, a proportional-integral-derivative (PID) controller provides decoupled feedback for the altitude, roll, and pitch channels, whereas the reinforcement learning policy
\[
    u_r\coloneqq\mu_\theta(x)
\]
is employed to regulate the yaw motion. The overall control input is therefore generated according to
\begin{align}
    u
    =
    \begin{bmatrix}
    T\\
    \tau_\phi\\
    \tau_\theta\\
    \tau_\psi
    \end{bmatrix}
    =
    \begin{bmatrix}
    \dfrac{m}{\mathbf{c}_\phi\mathbf{c}_\theta} & 0 & 0 & 0\\
    0 & I_{xx} & 0 & 0\\
    0 & 0 & I_{yy} & 0\\
    0 & 0 & 0 & I_{zz}
    \end{bmatrix}
    \begin{bmatrix}
    u_z\\
    u_\phi\\
    u_\theta\\
    u_r
    \end{bmatrix},
    \label{E27}
\end{align}
where $u_z$, $u_\phi$, and $u_\theta$ denote
the PID control actions for altitude, roll, and
pitch, respectively, while $u_r$ is the
reinforcement learning control input for yaw.
The altitude command includes gravity compensation
($u_z=g$ at hover), while the attitude commands
represent nominal angular accelerations near hover.

\section{Fault Modeling and Problem Formulation}

In practical unmanned aerial systems, actuator and sensor faults may arise due to component degradation, environmental disturbances, mechanical wear, or sensor drift. These faults can significantly deteriorate the control performance and may even compromise the stability of the closed-loop system if left uncompensated. Therefore, an accurate fault estimation mechanism is essential for achieving reliable fault-tolerant control.

This section introduces the fault models considered throughout this work and formulates the corresponding fault estimation problem. Unlike the conventional constant-fault assumption \cite{Wafi-CoDIT}, the proposed framework explicitly accounts for slowly time-varying faults, leading to a more realistic representation of practical operating conditions.

\subsection{Actuator and Sensor Fault Modeling}

Consider the discrete-time linearized system
\begin{align}
    \begin{aligned}
    x_{k+1}
    &=
    Ax_k+Bu_k+E_af_{a,k}+d_k,
    \label{eq:faulty_system}\\
    y_k
    &=
    Cx_k+E_sf_{s,k}.
    \end{aligned}
\end{align}
where $x_k\in\mathbb{R}^{n}$ denotes the system state, $u_k\in\mathbb{R}^{m}$ is the control input, and $y_k\in\mathbb{R}^{p}$ is the measured output. The vectors $f_{a,k}\in\mathbb{R}^{n_a}$ and $f_{s,k}\in\mathbb{R}^{n_s}$ are the actuator and sensor faults, respectively, while $d_k$ is an unknown disturbance. The matrices $E_a\in\mathbb{R}^{n\times n_a}$ and $E_s\in\mathbb{R}^{p\times n_s}$ specify the fault distribution in the actuator and sensor
channels.

The actuator fault $f_{a,k}$ enters the state
dynamics through $E_a$, whereas the sensor fault
$f_{s,k}$ directly corrupts the measured output.
Consequently, the estimation problem consists of reconstructing both fault signals from the available measurements while maintaining satisfactory closed-loop performance.

\subsection{Time-Varying Fault Model}

Most fault estimation methods assume that the actuator and sensor faults remain constant over the estimation interval. Such an assumption is appropriate for abrupt failures but may not accurately describe gradual degradation processes encountered in practical aerial vehicles. 
Let $f_k=[f_{a,k}^\top,f_{s,k}^\top]^\top$
and $\Delta_k=[\Delta_{a,k}^\top,
\Delta_{s,k}^\top]^\top$.
To capture these phenomena, the fault dynamics are modeled as
\begin{align}
    f_{k+1}
    =
    f_k+\Delta_k,
    \label{eq:time_varying_fault}
\end{align}
where $\Delta_k$ represents the fault variation between two consecutive sampling instants.

\begin{assumption}
\label{ass:slowly_varying_fault}
The actuator and sensor faults are assumed to vary slowly with time. Specifically,
\begin{align}\label{eq:slowly_varying_faults}
    \begin{aligned}
    f_{a,k+1}
    &=
    f_{a,k}
    +
    \Delta_{a,k},\\
    f_{s,k+1}
    &=
    f_{s,k}
    +
    \Delta_{s,k},
    \end{aligned}
\end{align}
where the fault variations satisfy $\|\Delta_{a,k}\|\le\delta_a$ and
$\|\Delta_{s,k}\|\le\delta_s$, for some positive constants $\delta_a,\delta_s>0$.
\end{assumption}

Assumption~\ref{ass:slowly_varying_fault} states that the actuator and sensor faults evolve sufficiently slowly compared with the system dynamics. 
The constant-fault model is recovered when $\delta_a=\delta_s=0$.

\subsection{Problem Formulation}

The objective of this work is to estimate the actuator and sensor faults while preserving the tracking performance of the closed-loop system. Specifically, given the faulty system \eqref{eq:faulty_system} satisfying Assumption~\ref{ass:slowly_varying_fault}, the goal is to design

\begin{enumerate}
    \item a proportional--integral observer capable of estimating slowly time-varying actuator and sensor faults;
    \item a fault compensation strategy based on the estimated fault signals; and
    \item a learning-oriented controller that maintains satisfactory tracking performance in the presence of bounded time-varying faults.
\end{enumerate}

The observer design and its theoretical properties are developed in the following section.

\section{PI Observer for Time-Varying Fault Estimation}
\label{sec:PI_observer}

This section develops a proportional--integral (PI) observer for estimating actuator and sensor faults. In contrast to the constant-fault setting, the analysis explicitly accounts for slowly time-varying faults. The main result shows that the estimation error converges to zero for constant faults and remains uniformly bounded when the fault variation is bounded.

\subsection{PI Observer for Actuator Faults}

We first consider the actuator-fault model
\begin{align}
    \begin{aligned}
    x_{k+1}&=Ax_k+Bu_k+E_af_{a,k},\\
    y_k&=Cx_k,
    \end{aligned}
    \label{eq:actuator_fault_system}
\end{align}
where $x_k\in\mathbb{R}^n$ is the state, $u_k\in\mathbb{R}^m$ is the control input, $y_k\in\mathbb{R}^p$ is the measured output, and $f_{a,k}\in\mathbb{R}^r$ denotes the actuator fault. The matrix $E_a\in\mathbb{R}^{n\times r}$ specifies the actuator fault distribution. 

For clarity of the observer derivation, we temporarily omit the unknown disturbance term $d_k$ and consider actuator faults only. 
Bounded disturbances can be incorporated as additional
inputs to the estimation error dynamics, although
they generally introduce a nonzero estimation-error
bound.

Under Assumption~\ref{ass:slowly_varying_fault}, the actuator fault satisfies
\begin{align*}
    f_{a,k+1}
    =
    f_{a,k}
    +
    \Delta_{a,k},
\end{align*}
where $\Delta_{a,k}\in\mathbb{R}^r$ denotes the actuator fault variation between two consecutive sampling instants. The constant-fault case is recovered by setting $\Delta_{a,k}=0$.

Define the augmented state
\begin{align*}
    \bar{x}_k
    =
    \begin{bmatrix}
        x_k\\
        f_{a,k}
    \end{bmatrix}.
\end{align*}
Then, the augmented system is
\begin{align}
    \bar{x}_{k+1}
    =
    \bar{A}\bar{x}_k
    +
    \bar{B}u_k
    +
    \bar{I}\Delta_{a,k},
    \qquad
    y_k
    =
    \bar{C}\bar{x}_k,
    \label{eq:augmented_actuator_system}
\end{align}
where
\begin{align*}
    \bar{A}
    =
    \begin{bmatrix}
        A & E_a\\
        0 & I_r
    \end{bmatrix},
    \qquad
    \bar{B}
    =
    \begin{bmatrix}
        B\\
        0
    \end{bmatrix},
    \qquad
    \bar{C}
    =
    \begin{bmatrix}
        C & 0
    \end{bmatrix},
    \qquad
    \bar{I}
    =
    \begin{bmatrix}
        0_{n\times r}\\
        I_r
    \end{bmatrix}.
\end{align*}

\begin{assumption}
\label{ass:actuator_detectability}
The augmented pair $(\bar{A},\bar{C})$ is detectable.
\end{assumption}

A proportional--integral (PI) observer for the augmented system \eqref{eq:augmented_actuator_system} is designed as
\begin{align}
    \begin{aligned}
    \hat{x}_{k+1}
    &=
    (A-L_PC)\hat{x}_k
    +
    Bu_k
    +
    E_a\hat{f}_{a,k}
    +
    L_Py_k,
    \\
    \hat{f}_{a,k+1}
    &=
    \hat{f}_{a,k}
    +
    L_I(y_k-C\hat{x}_k),
    \end{aligned}
    \label{eq:pi_observer_actuator}
\end{align}
where $L_P\in\mathbb{R}^{n\times p}$ and $L_I\in\mathbb{R}^{r\times p}$ are the proportional and integral observer gains, respectively. Equivalently, defining
\[
    \hat{\bar{x}}_k
    =
    \begin{bmatrix}
    \hat{x}_k\\
    \hat{f}_{a,k}
    \end{bmatrix},
    \qquad
    \bar{L}
    =
    \begin{bmatrix}
    L_P\\
    L_I
    \end{bmatrix},
 \]
the observer can be written compactly as
\begin{align}
    \hat{\bar{x}}_{k+1}
    =
    \bar{A}\hat{\bar{x}}_k
    +
    \bar{B}u_k
    +
    \bar{L}
    \bigl(
    y_k-\bar{C}\hat{\bar{x}}_k
    \bigr).
    \label{eq:compact_pi_observer}
\end{align}
Let
\begin{align}
    e_k
    =
    \bar{x}_k-\hat{\bar{x}}_k
    =
    \begin{bmatrix}
    e_{x,k}\\
    e_{f,k}
    \end{bmatrix}
    =
    \begin{bmatrix}
    x_k-\hat{x}_k\\
    f_{a,k}-\hat{f}_{a,k}
    \end{bmatrix}.
\end{align}
Subtracting \eqref{eq:compact_pi_observer} from \eqref{eq:augmented_actuator_system} gives the estimation error dynamics
\begin{align}
    e_{k+1}
    =
    \left(
    \bar{A}-\bar{L}\bar{C}
    \right)e_k
    +
    \bar{I}\Delta_{a,k}.
    \label{eq:actuator_error_dynamics}
\end{align}
Equation~\eqref{eq:actuator_error_dynamics} shows that the actuator fault variation $\Delta_{a,k}$ enters the estimation error dynamics as an external input. This observation forms the basis for the robustness analysis developed in the following subsections.

\subsection{Observability Condition}
\label{subsec:observability_condition}

The proportional--integral observer introduced in the
previous subsection requires detectability of the
augmented pair $(\bar{A},\bar{C})$.
The following result provides a necessary and sufficient
condition for observability, which is sufficient for
observer design.

\begin{lemma}
\label{lem:actuator_observability}
Suppose that the pair $(A,C)$ is observable. Then the augmented pair $(\bar{A},\bar{C})$ is observable if and only if
\begin{align}
    \rank
    \begin{bmatrix}
        A-I_n & E_a\\
        C & 0
    \end{bmatrix}
    =
    n+r.
    \label{eq:actuator_observability_condition}
\end{align}
\end{lemma}

\begin{proof}
According to the Popov--Belevitch--Hautus (PBH) observability test, the pair $(\bar{A},\bar{C})$ is observable if and only if
\begin{align}
    \rank
    \begin{bmatrix}
        \lambda I_{n+r}-\bar{A}\\
        \bar{C}
    \end{bmatrix}
    =
    n+r
    \label{eq:pbh_augmented_observability}
\end{align}
for every eigenvalue $\lambda$ of $\bar{A}$. Since
\[
    \bar{A}
    =
    \begin{bmatrix}
    A&E_a\\
    0&I_r
    \end{bmatrix},
\]
the spectrum of $\bar{A}$ consists of the eigenvalues of $A$ together with the eigenvalue $\lambda=1$ introduced by the fault dynamics.
For $\lambda\neq1$, the PBH test reduces to the observability of $(A,C)$, which holds by assumption. For $\lambda=1$, the PBH matrix becomes
\begin{align}
    \begin{bmatrix}
        I_n-A & -E_a\\
        0 & 0\\
        C & 0
    \end{bmatrix},
    \label{eq:pbh_lambda_one}
\end{align}
which has full column rank if and only if
\[
    \rank
    \begin{bmatrix}
    A-I_n & E_a\\
    C & 0
    \end{bmatrix}
    =
    n+r.
\]
Therefore, the stated rank condition guarantees observability of $(\bar{A},\bar{C})$.
\end{proof}

Under this condition, the gain $\bar{L}$ can be chosen so that
$A_e\coloneqq\bar{A}-\bar{L}\bar{C}$
is Schur stable, namely all eigenvalues of $A_e$ lie strictly inside the unit circle.

\subsection{Estimation Error Bound for Time-Varying Faults}

The observer error dynamics established in the previous subsection enable us to characterize the effect of slowly time-varying faults on the estimation performance. The following theorem shows that the observer remains robust with respect to bounded fault variations and provides an explicit bound on the estimation error.

\begin{theorem}
\label{thm:actuator_error_bound}

Suppose that $(\bar{A},\bar{C})$ is observable and that the observer gain $\bar{L}$ is chosen such that
$A_e=\bar{A}-\bar{L}\bar{C}$
is Schur stable. Assume further that the actuator fault variation satisfies
$\|\Delta_{a,k}\|\le\delta$, $\forall k\ge 0,$
for some $\delta\ge 0$. Then the estimation error system
\eqref{eq:actuator_error_dynamics}
is input-to-state stable with respect to $\Delta_{a,k}$. In particular, there exist constants $M\ge1$ and $0<\rho<1$ such that
\begin{subequations}
\begin{align}
    \|e_k\|
    \le
    M\rho^k\|e_0\|
    +
    \frac{M\|\bar{I}\|}{1-\rho}\delta,
    \qquad
    \forall k\ge0.
    \label{eq:actuator_error_bound}
\end{align}
Consequently, the estimation error is ultimately bounded as
\begin{align}
    \limsup_{k\to\infty}\|e_k\|
    \le
    \frac{M\|\bar{I}\|}{1-\rho}\delta.
    \label{eq:actuator_ultimate_bound}
\end{align}
\end{subequations}
\end{theorem}

\begin{proof}
From \eqref{eq:actuator_error_dynamics}, the solution of the error dynamics is
\begin{align}
    e_k
    =
    A_e^k e_0
    +
    \sum_{j=0}^{k-1}
    A_e^{k-1-j}
    \bar{I}\Delta_{a,j}.
    \label{eq:error_solution}
\end{align}
Since $A_e$ is Schur stable, there exist constants $M\ge1$ and $0<\rho<1$ such that
$\|A_e^k\|\le M\rho^k,\forall k\ge0.$
Taking norms in \eqref{eq:error_solution} and using the assumption
$\|\Delta_{a,k}\|\le\delta$ gives
\begin{align}
    \|e_k\|
    &\le
    \|A_e^k\|\|e_0\|
    +
    \sum_{j=0}^{k-1}
    \|A_e^{k-1-j}\|
    \|\bar{I}\|
    \|\Delta_{a,j}\|
    \nonumber\\
    &\le
    M\rho^k\|e_0\|
    +
    M\|\bar{I}\|\delta
    \sum_{j=0}^{k-1}
    \rho^{k-1-j}.
    \label{eq:error_bound_intermediate}
\end{align}
Using the geometric-series bound
\[
    \sum_{j=0}^{k-1}\rho^{k-1-j}
    =
    \sum_{\ell=0}^{k-1}\rho^\ell
    \le
    \frac{1}{1-\rho},
\]
we obtain
\[
    \|e_k\|
    \le
    M\rho^k\|e_0\|
    +
    \frac{M\|\bar{I}\|}{1-\rho}\delta,
\]
which proves \eqref{eq:actuator_error_bound}. Taking the limit superior as $k\to\infty$ immediately yields
\eqref{eq:actuator_ultimate_bound}.
\end{proof}

\begin{corollary}
\label{cor:constant_fault}
If the actuator fault is constant, that is, $\Delta_{a,k}=0$ for all $k\ge0$, then the observer error converges exponentially to the origin. In particular,
\begin{align}
    \|e_k\|
    \le
    M\rho^k\|e_0\|,
    \qquad
    \lim_{k\to\infty}e_k=0.
    \label{eq:constant_fault_convergence}
\end{align}
\end{corollary}

\begin{proof}
Setting $\Delta_{a,k}=0$ in Theorem~\ref{thm:actuator_error_bound} eliminates the disturbance term in \eqref{eq:actuator_error_bound}, yielding $\|e_k\|\le M\rho^k\|e_0\|.$
Since $0<\rho<1$, it follows that $\rho^k\to0$ as $k\to\infty$. Therefore,
\[
    \lim_{k\to\infty}e_k=0,
\]
which establishes exponential convergence of both the state and actuator-fault estimation errors.
\end{proof}

\begin{corollary}
\label{cor:small_sampling}
Suppose that the constants $M$ and $\rho$ in
Theorem~\ref{thm:actuator_error_bound} can be
chosen independently of the sampling period $T$.
If the actuator fault satisfies
\[
    \Delta_{a,k}=f_{a,k+1}-f_{a,k}=O(T^2)
\]
uniformly in $k$, then the estimation error is
ultimately bounded by an $O(T^2)$ neighborhood
of the origin.
\end{corollary}

\begin{proof}
If $\|\Delta_{a,k}\|\le cT^2$ for some constant $c>0$, then Theorem~\ref{thm:actuator_error_bound} implies
\[
    \limsup_{k\to\infty}\|e_k\|
    \le
    \frac{M\|\bar{I}\|}{1-\rho}cT^2
    =
    O(T^2),
\]
which completes the proof.
\end{proof}

\subsection{Combined Actuator and Sensor Faults}

The preceding development can be extended to simultaneous actuator and sensor faults. Consider
\begin{align}
    \begin{aligned}
    x_{k+1}
    &=
    Ax_k+Bu_k+E_af_{a,k}+d_k,\\
    y_k
    &=
    Cx_k+E_sf_{s,k},
    \end{aligned}
    \label{eq:combined_fault_system}
\end{align}
where $f_{a,k}$ and $f_{s,k}$ denote actuator and sensor faults, respectively, while $d_k$ denotes an unknown bounded disturbance. Define the combined fault vector
\begin{align*}
    f_k=
    \begin{bmatrix}
        f_{a,k}\\
        f_{s,k}
    \end{bmatrix}
    \in\mathbb{R}^r.
\end{align*}
Here, $r=n_a+n_s$, where $n_a$ and $n_s$
denote the numbers of modeled actuator and sensor
fault channels, respectively. In the numerical
experiments, four fault channels are considered,
with their allocation depending on the fault scenario.

To incorporate the sensor fault into the state equation, an auxiliary variable $z_k$ is introduced as
\begin{align}
    z_{k+1}
    =
    M(y_k-z_k)
    =
    M(Cx_k+E_sf_{s,k}-z_k),
    \label{eq:auxiliary_state}
\end{align}
where $M\in\mathbb{R}^{p\times p}$ is chosen
to be nonsingular, with $-M$ Schur stable.
Defining
\[
    \bar{x}_k
    =
    \begin{bmatrix}
        x_k\\
        z_k
    \end{bmatrix},
\]
one obtains
\begin{align}
    \begin{aligned}
    \bar{x}_{k+1}
    &=A_a\bar{x}_k+B_a u_k+D_ad_k+E_a^\star f_k, \\
    z_k
    &=C_a\bar{x}_k,
    \end{aligned}
    \label{eq:combined_augmented_system}
\end{align}
with
\begin{align*}
    A_a
    =
    \begin{bmatrix}
        A & 0\\
        MC & -M
    \end{bmatrix},
    \qquad
    B_a
    =
    \begin{bmatrix}
        B\\
        0
    \end{bmatrix},
    \qquad
    D_a =
    \begin{bmatrix}
        I_n\\
        0_{p\times n}
    \end{bmatrix},
    \qquad
    E_a^\star
    =
    \begin{bmatrix}
        E_a & 0\\
        0 & ME_s
    \end{bmatrix},
    \qquad
    C_a
    =
    \begin{bmatrix}
        0 & I
    \end{bmatrix}.
\end{align*}
Assuming the combined fault satisfies $f_{k+1}=f_k+\Delta_k,$ and $\|\Delta_k\|\le\delta,$
we define the extended state
\begin{align*}
    \tilde{x}_k
    =
    \begin{bmatrix}
        \bar{x}_k\\
        f_k
    \end{bmatrix}.
\end{align*}
Then
\begin{align}
    \begin{aligned}
    \tilde{x}_{k+1}
    &=\tilde{A}\tilde{x}_k
    +\tilde{B}u_k
    +\tilde{D}d_k
    +\tilde{I}\Delta_k,\\
    z_k
    &=\tilde{C}\tilde{x}_k,
    \end{aligned}
    \label{eq:combined_extended_system}
\end{align}
where
\begin{align*}
    \tilde{A}
    =
    \begin{bmatrix}
        A_a & E_a^\star\\
        0 & I_r
    \end{bmatrix},
    \qquad
    \tilde{B}
    =
    \begin{bmatrix}
        B_a\\
        0
    \end{bmatrix},
    \qquad
    \tilde{D}
    =
    \begin{bmatrix}
        D_a\\
        0
    \end{bmatrix},
    \qquad
    \tilde{C}
    =
    \begin{bmatrix}
        C_a & 0
    \end{bmatrix},
    \qquad
    \tilde{I}
    =
    \begin{bmatrix}
        0\\
        I_r
    \end{bmatrix}.
\end{align*}
The augmented pair $(\tilde{A},\tilde{C})$
is assumed to be detectable, so that an
observer gain $\tilde{L}$ can be selected
such that $\tilde{A}-\tilde{L}\tilde{C}$
is Schur stable.
In particular, a necessary condition is
\begin{align}
    \rank
    \begin{bmatrix}
        A-I_n & E_a & 0\\
        C & 0 & E_s
    \end{bmatrix}
    =n+r,
    \label{eq:combined_fault_rank}
\end{align}
which must be verified for the selected fault channels.

A PI observer can then be constructed as
\begin{align}
    \hat{\tilde{x}}_{k+1}
    &=
    \tilde{A}\hat{\tilde{x}}_k
    +
    \tilde{B}u_k
    +
    \tilde{L}
    \bigl(
        z_k-\tilde{C}\hat{\tilde{x}}_k
    \bigr),
    \label{eq:combined_pi_observer}
\end{align}
where $\tilde{L}$ is chosen such that $\tilde{A}-\tilde{L}\tilde{C}$
is Schur stable. The corresponding estimation error satisfies
\begin{align}
    \tilde{e}_{k+1}
    =
    (\tilde{A}-\tilde{L}\tilde{C})\tilde{e}_k
    +
    \tilde{D}d_k
    +
    \tilde{I}\Delta_k.
    \label{eq:combined_error_dynamics}
\end{align}
Equation~\eqref{eq:combined_error_dynamics} shows that both the disturbance $d_k$ and the fault variation $\Delta_k$ act as external inputs to the observer error dynamics.

This result shows that the PI observer provides exact asymptotic reconstruction for constant actuator and sensor faults in the absence of disturbances, and practical reconstruction for bounded disturbances and slowly time-varying faults. The estimation accuracy improves as the disturbance and fault variation bounds decrease.

\begin{corollary}\label{cor:combined_fault_bound}
Consider the combined actuator--sensor fault model in \eqref{eq:combined_extended_system}
and the PI observer in \eqref{eq:combined_pi_observer}.
Suppose that the gain $\tilde{L}$ is chosen such that
$\tilde{A}-\tilde{L}\tilde{C}$ is Schur stable. Furthermore, assume that
\[
    \|d_k\|\le d_{\max},
    \qquad
    \|\Delta_k\|\le\delta,
    \qquad
    \forall k\ge0.
\]
Then the combined actuator--sensor estimation error is ultimately bounded. In particular, there exist constants
$\tilde{M}\ge1$
and
$0<\tilde{\rho}<1$
such that
\begin{align}
    \|\tilde{e}_k\|
    \le
    \tilde{M}\tilde{\rho}^k
    \|\tilde{e}_0\|
    +
    \frac{\tilde{M}}{1-\tilde{\rho}}
    \left(
        \|\tilde{D}\|d_{\max}
        +
        \|\tilde{I}\|\delta
    \right),
    \qquad
    \forall k\ge0.
    \label{eq:combined_error_bound}
\end{align}
\end{corollary}

\begin{proof}
From \eqref{eq:combined_error_dynamics},
\[
    \tilde{e}_{k+1}
    =
    (\tilde{A}-\tilde{L}\tilde{C})\tilde{e}_k
    +
    \tilde{D}d_k
    +
    \tilde{I}\Delta_k.
\]
Since $\tilde{A}-\tilde{L}\tilde{C}$ is Schur stable, the proof follows exactly as in Theorem~\ref{thm:actuator_error_bound} by treating
$\tilde{D}d_k+\tilde{I}\Delta_k$
as the bounded input to the error system.
Applying the same geometric-series argument yields
\eqref{eq:combined_error_bound}.
\end{proof}

\section{Fault-Tolerant Control}

\begin{figure}[h!]
    \centering
    \includegraphics[width=.6\linewidth]{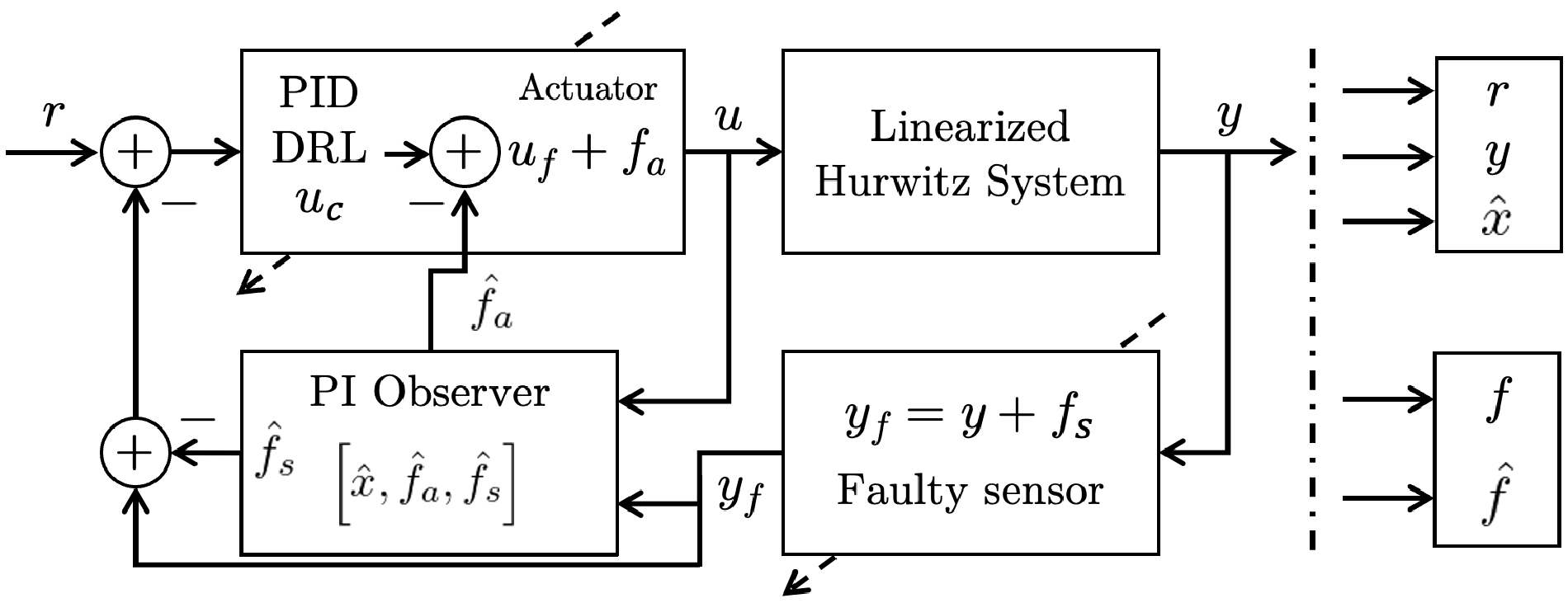}
    \caption{Fault-tolerant control scheme with faults $f=[f_a,f_s]^\top$ using PI Observer.}
    \label{F2}
\end{figure}

The fault estimation framework developed in the previous section is now integrated into the control architecture to compensate for actuator and sensor faults online. Rather than redesigning the controller after a fault occurs, the proposed strategy reconstructs the fault signals using the proportional--integral observer and compensates them before they propagate through the closed-loop system. The resulting architecture is illustrated in Fig.~\ref{F2}.

Let $\hat{f}_{a,k}$ and $\hat{f}_{s,k}$ denote the actuator and sensor fault estimates obtained from the PI observer. The measured output is first corrected according to
\begin{align}
    y_{c,k}
    =
    y_k
    -
    E_s\hat{f}_{s,k},
    \label{eq:corrected_output}
\end{align}
where $y_{c,k}$ denotes the compensated measurement supplied to the feedback controller.
For actuator fault compensation, we assume
matched additive faults, namely $E_a=B$,
with $f_{a,k}\in\mathbb{R}^m$.
Similarly, the actuator fault estimate is incorporated through the compensated control input
\begin{align}
    u_{c,k}
    =
    u_k
    -
    \hat{f}_{a,k},
    \label{eq:compensated_input}
\end{align}
which is transmitted to the UAV actuators.

Consequently, the fault-tolerant closed-loop system becomes
\begin{subequations}
\begin{align}
    x_{k+1}
    =
    Ax_k+Bu_k
    +B(f_{a,k}-\hat{f}_{a,k})
    +d_k,
    \label{eq:ftc_system}
\end{align}
with corrected measurement
\begin{align}
    y_{c,k}
    =
    Cx_k
    +
    E_s
    \Big[
        f_{s,k}-\hat{f}_{s,k}
    \Big].
    \label{eq:ftc_output}
\end{align}
\end{subequations}
When the observer converges, 
$\hat{f}_{a,k}\rightarrow f_{a,k}$ and $\hat{f}_{s,k}\rightarrow f_{s,k},$
the effects of both actuator and sensor faults are asymptotically removed from the closed-loop system. For slowly varying faults, Theorem~\ref{thm:actuator_error_bound} and Corollary~\ref{cor:combined_fault_bound} imply that the residual fault
\[
    f_k-\hat{f}_k
\]
remains uniformly bounded. 
Consequently, the residual fault effects remain
bounded under the conditions of the observer
stability analysis. Their influence on closed-loop
tracking performance depends on the robustness
properties of the feedback controller.

The compensated measurement $y_{c,k}$ and compensated control input $u_{c,k}$ are subsequently employed by the learning-oriented controller developed in the following section.

\section{Learning-Oriented Control}

\begin{figure}[h!]
    \centering
    \includegraphics[width=.5\linewidth]{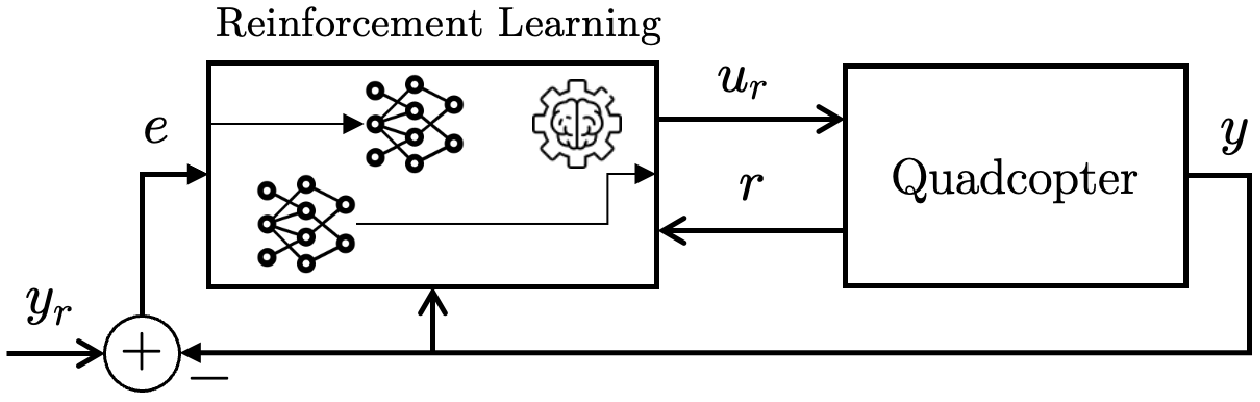}
    \caption{Scheme of training for DRL on single state $y_4=\psi$.}
    \label{F3}
\end{figure}

To compensate for nonlinear uncertainties and improve the closed-loop tracking performance, a Deep Deterministic Policy Gradient (DDPG) controller is adopted. Unlike stochastic policy gradient methods, DDPG employs a deterministic policy
\begin{align}
    u=\mu_\theta(x),
    \label{eq:deterministic_policy}
\end{align}
which directly maps the observed state $x\in S_X$ to the continuous control action $u\in S_A$. Here, $\theta$ denotes the parameter vector of the actor network.
The objective is to maximize the expected discounted return
\begin{align}
    J(\mu_\theta)
    =
    \mathbb{E}\!\left[r_t^\gamma\,|\,\mu_\theta\right],
    \label{eq:expected_return}
\end{align}
where
\begin{align*}
    J(\mu_\theta)
    =
    \int_{S_X}
    \rho^\mu(x)
    r\bigl(x,\mu_\theta(x)\bigr)\,dx,
\end{align*}
with $\rho^\mu(x)$ denoting the unnormalized
discounted state-occupancy distribution.

The cumulative discounted reward is defined as
\begin{align}
    r_t^\gamma
    =
    \sum_{k=t}^{\infty}
    \gamma^{k-t}
    r(x_k,u_k),
    \label{eq:discounted_reward}
\end{align}
where $\gamma\in(0,1)$ is the discount factor. The policy parameters are updated by gradient ascent according to
\begin{align}
    \theta
    \leftarrow
    \theta
    +
    \alpha_\ell
    \nabla_\theta
    J(\mu_\theta),
    \label{eq:policy_update}
\end{align}
where $\alpha_\ell>0$ denotes the learning rate.

Direct evaluation of $\nabla_\theta J(\mu_\theta)$ is computationally demanding because it depends on the unknown action-value function. To overcome this difficulty, DDPG employs an actor--critic architecture consisting of an actor network $\mu_\theta$ and a critic network $Q^\vartheta$, where $\vartheta$ denotes the critic parameters. The deterministic policy gradient theorem gives the
policy gradient in terms of the true action-value
function $Q^\mu$, which is approximated by the
critic network $Q^\vartheta$ in DDPG:
\begin{align}
    \nabla_\theta J(\mu_\theta)
    =
    \int_{S_X}
    \rho^\mu(x)
    \nabla_\theta\mu_\theta(x)
    \nabla_u Q^\mu(x,u)
    \Big|_{u=\mu_\theta(x)}
    dx,
    \label{eq:deterministic_policy_gradient}
\end{align}

Accordingly, the actor parameters are updated as
\begin{align}
    \theta
    \leftarrow
    \theta
    +
    \alpha_\ell
    \nabla_\theta\mu_\theta(x)
    \nabla_uQ^\vartheta(x,u)
    \Big|_{u=\mu_\theta(x)}.
    \label{eq:actor_update}
\end{align}

The overall DDPG training architecture adopted in this
work is illustrated in Fig.~\ref{F3}. The reinforcement
learning controller is designed for the yaw channel
($y_4=\psi$), while the altitude, roll, and pitch motions
are regulated by the nominal controller. The DDPG agent
is trained without actuator or sensor faults. During
implementation, the trained policy is integrated with
the PI-observer-based fault-tolerant control architecture
to compensate for actuator and sensor faults.

The observation vector is chosen as
\begin{align}
    x_r
    =
    \begin{bmatrix}
        y_4 &
        e &
        e_i
    \end{bmatrix}^{\!\top},
    \label{eq:observation_vector}
\end{align}
where $e$ denotes the yaw tracking error and $e_i$ is its integral. 
The upper and lower bounds of the observation space are chosen as
$S_{X_u}=[\pi,\infty,\infty]^\top$ and $S_{X_l}=[-\pi,-\infty,-\infty]^\top$, respectively.
The reward function is defined as
\begin{align}
    r_t
    =
    \begin{cases}
        10, & |e|<\varepsilon,\\
        1, & \text{otherwise},
    \end{cases}
    \label{eq:reward_function}
\end{align}
where $\varepsilon>0$ is a prescribed tracking tolerance. 
Consequently, the agent is encouraged to keep
the yaw tracking error within the prescribed
tolerance band during training.

\section{Numerical Examples and Simulations}

This section presents numerical simulations to validate the proposed proportional--integral (PI) observer and the fault-tolerant learning-oriented control framework. The simulations consist of two parts. First, the estimation performance of the proposed PI observer is evaluated under various actuator and sensor fault scenarios. Second, the estimated faults are incorporated into the fault-tolerant controller, where a hybrid PID--DDPG strategy is employed to compensate for the faults and improve the tracking performance.

The UAV parameters are chosen as follows: the mass is $m=0.468~\mathrm{kg}$, the translational drag coefficients are $A_x=0.3$, $A_y=0.3$, and $A_z=0.25$, the rotational drag coefficient is $A_r=0.2$, and the inertia matrix is characterized by
\[
    I_{xx}=4.856\times10^{-3},
    \qquad
    I_{yy}=I_{xx},
    \qquad
    I_{zz}=8.801\times10^{-3}.
\]
The continuous-time linearized subsystem is
\begin{align*}
    A_c
    =
    \begin{bmatrix}
        0_{4\times4} & I_4\\
        0_{4\times4} &
        \diag(\Psi)\diag(\Psi_A)
    \end{bmatrix},\qquad
    B_c
    =
    \begin{bmatrix}
        0_{4\times4}\\
        \diag(\Psi)
    \end{bmatrix}
\end{align*}
where
\[
    \Psi
    =
    \begin{bmatrix}
    m^{-1} &
    I_{xx}^{-1} &
    I_{yy}^{-1} &
    I_{zz}^{-1}
    \end{bmatrix},
    \]
    and
    \[
    \Psi_A
    =
    \begin{bmatrix}
    A_z &
    A_r &
    A_r &
    A_r
    \end{bmatrix}^{\!\top}.
\]
The discrete-time matrices $A$ and $B$ used
in the observer and controller are obtained
from $(A_c,B_c)$ using the sampling period $T$.

The controller combines a conventional PID controller with the proposed DDPG policy. Only the yaw state $y_4$ is used for reinforcement learning, whereas the learned policy is subsequently applied to the remaining attitude channels $y_2$ and $y_3$. The DDPG agent is trained for $100~\mathrm{s}$ over $51$ episodes with randomly generated reference trajectories. The observation vector is given by $x_r=[y_4,e,e_i]^\top$, where $e$ denotes the tracking error and $e_i$ its integral. The reward is selected as $r_k=10$ whenever $|e|<\varepsilon$, with $\varepsilon=0.01$, and $r_k=1$ otherwise. The actor and critic learning rates are chosen as $\alpha_a=10^{-4}$ and $\alpha_c=10^{-3}$, respectively, with discount factor $\gamma=0.9$, mini-batch size $64$, and maximum episode and step numbers equal to $1000$.

To evaluate the proposed observer and fault-tolerant controller, three representative fault scenarios are considered:

\begin{enumerate}
    \item Four simultaneous actuator faults occurring at $t=50~\mathrm{s}$ with loss of effectiveness $\alpha=0.5$ and bias
    $\beta=[2,0.3,0.15,0.1]$.
    
    \item Four actuator faults occurring sequentially at
    $t=[50,60,70,80]~\mathrm{s}$ with identical bias $\beta$ and loss of effectiveness $\alpha=0.75$.
    
    \item Two actuator faults introduced at $t=50~\mathrm{s}$ followed by two sensor faults at $t=[60,70]~\mathrm{s}$ with $\alpha=1$ and
    $\beta=[5,0.3,0.15,0.01]$, corresponding to
    $f=[f_a,f_a,f_s,f_s]^\top$.
\end{enumerate}

Figure~\ref{f4a} illustrates the influence of actuator loss of effectiveness on the UAV dynamics, whereas Fig.~\ref{f4b} compares the tracking performance of the conventional PID controller and the proposed DDPG controller for the three attitude states $y_2$, $y_3$, and $y_4$. The DDPG controller achieves improved tracking accuracy and noticeably smoother responses despite being trained only on the yaw channel.

The estimation results of the proposed PI observer are first presented in Fig.~\ref{f4c}--Fig.~\ref{f4d}. These simulations illustrate the observer's
fault-estimation performance under the considered
actuator-only and combined actuator--sensor
fault scenarios. The estimated fault signals are subsequently used to compensate the corrupted actuator inputs and sensor measurements according to the fault-tolerant control strategy developed in Section~5.

Finally, the complete fault-tolerant learning-oriented control framework is evaluated using the estimated faults obtained from the PI observer. Figure~\ref{f5} shows that the proposed controller successfully maintains accurate trajectory tracking under all three fault scenarios. 
These numerical results illustrate the tracking
performance of the integrated PI-observer-based
FTC and hybrid PID--DDPG controller under the
considered fault scenarios.

The additional PI-observer simulation is carried out for two time-varying actuator faults under bounded disturbances. In this case, the faulty actuator channels are selected as $1$ and $4$, namely
\[
    f_{a,k}
    =
    \begin{bmatrix}
        f_{a,1,k}\\
        f_{a,4,k}
    \end{bmatrix}.
\]
The time-varying actuator faults are chosen as
\[
    f_{a,1}(t)
    =
    0.30+0.8\sin(0.8t),
    \qquad
    f_{a,4}(t)
    =
    0.55+0.4\cos(0.5t).
\]
The bounded disturbance is selected as
\[
    d(t)
    =
    0.5
    \begin{bmatrix}
        \sin(0.15t) &
        \cos(0.12t) &
        \sin(0.10t) &
        \cos(0.08t) &
        \sin(0.20t) &
        \cos(0.18t) &
        \sin(0.16t) &
        \cos(0.14t)
    \end{bmatrix}^{\!\top}.
\]
The corresponding actuator fault variation is computed as
\[
    \Delta_{a,k}
    =
    f_{a,k+1}-f_{a,k}.
\]
The observer gain is obtained by pole placement with poles uniformly distributed in the interval $[0.2,0.75]$. The estimation performance is shown in Fig.~\ref{fig:time_varying_actuator_faults}, where the PI observer accurately tracks the time-varying actuator faults despite the presence of bounded disturbances.

Experimental validation on the Quanser QDrone~2 autonomous aerial vehicle is considered as future work. The QDrone 2 is an open-architecture research-grade drone that features a robust on-board NVIDIA Jetson Xavier NX system-on-module, multiple high-resolution cameras, and built-in WiFi, making it suitable for advanced research applications.

\begin{figure}[t!]
    \centering
    \subfloat[\label{f4a}\scriptsize System response under different factors $\alpha$.]{
    \includegraphics[width=.33\linewidth]{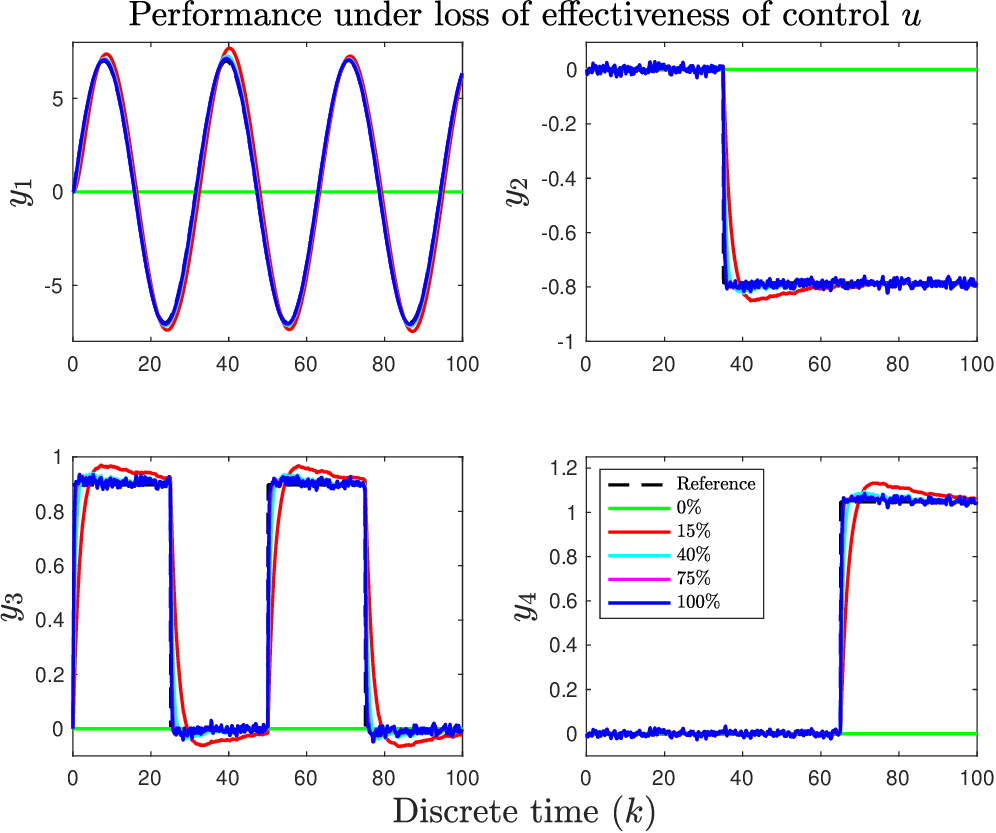}}\qquad
    \subfloat[\label{f4b}\scriptsize Tracking comparison: PID and DDPG controllers.]{
    \includegraphics[width=.335\linewidth]{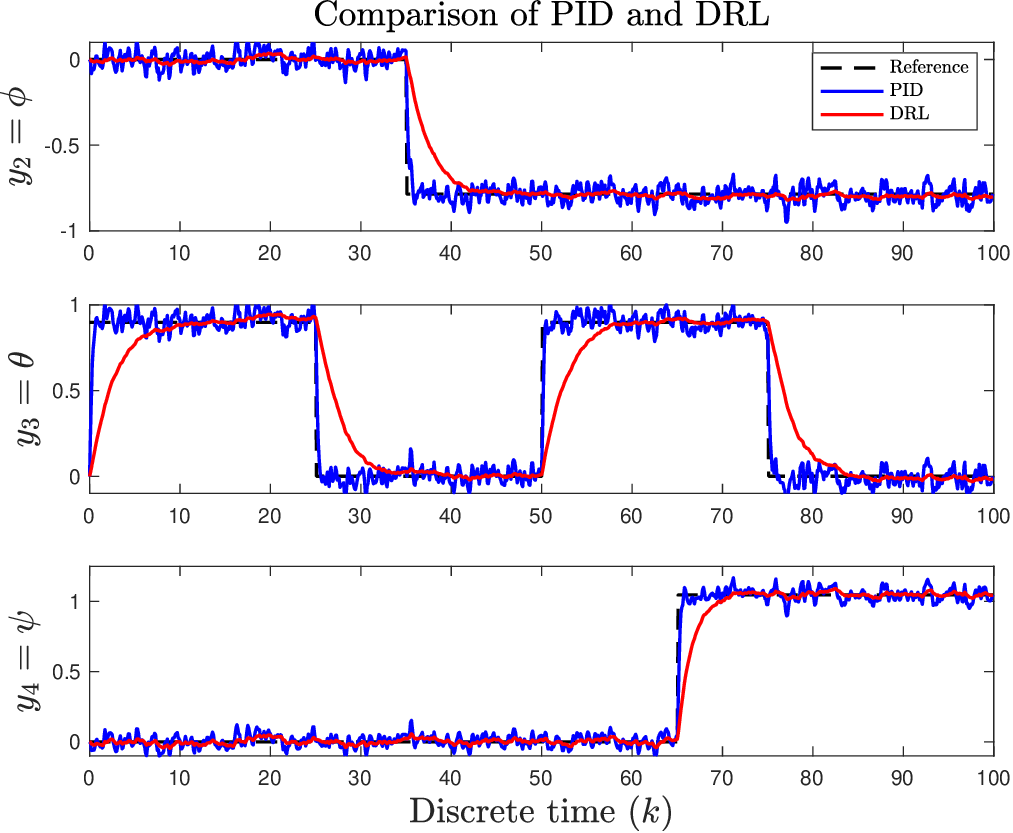}}\\
    \subfloat[\label{f4c}\scriptsize Four faults with different occurrence times.]{
    \includegraphics[width=.33\linewidth]{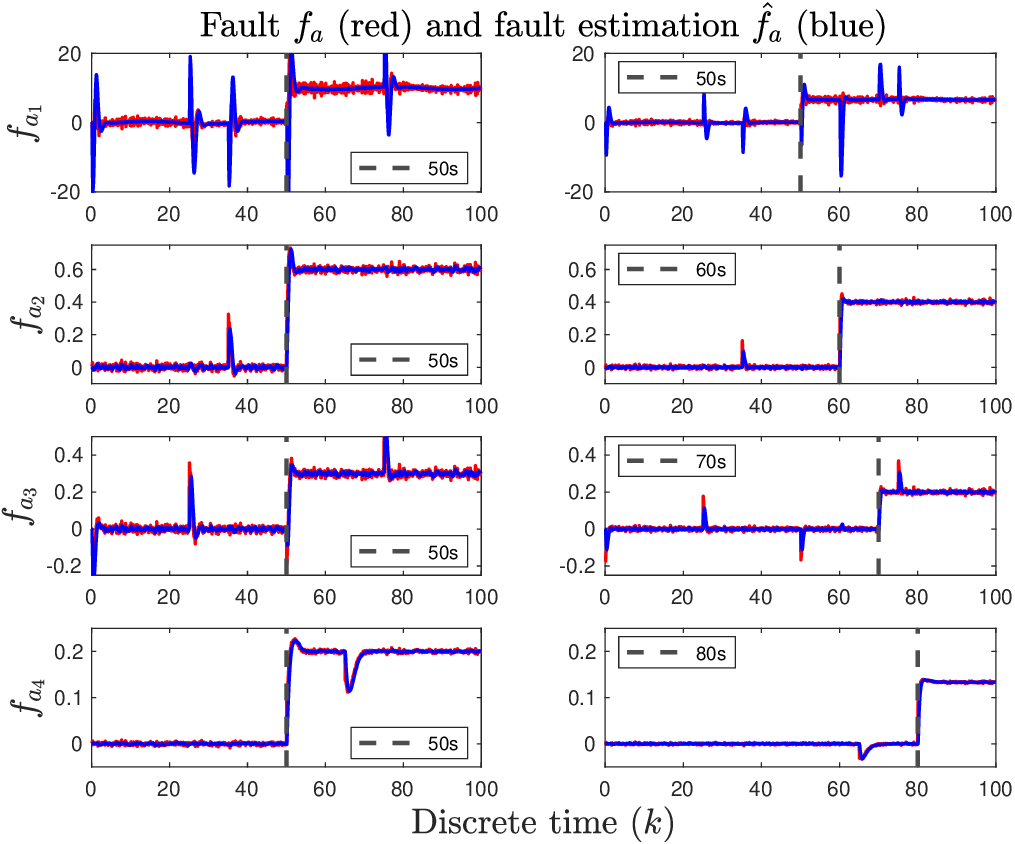}}\qquad
    \subfloat[\label{f4d}\scriptsize Two actuator and two sensor faults.]{
    \includegraphics[width=.33\linewidth]{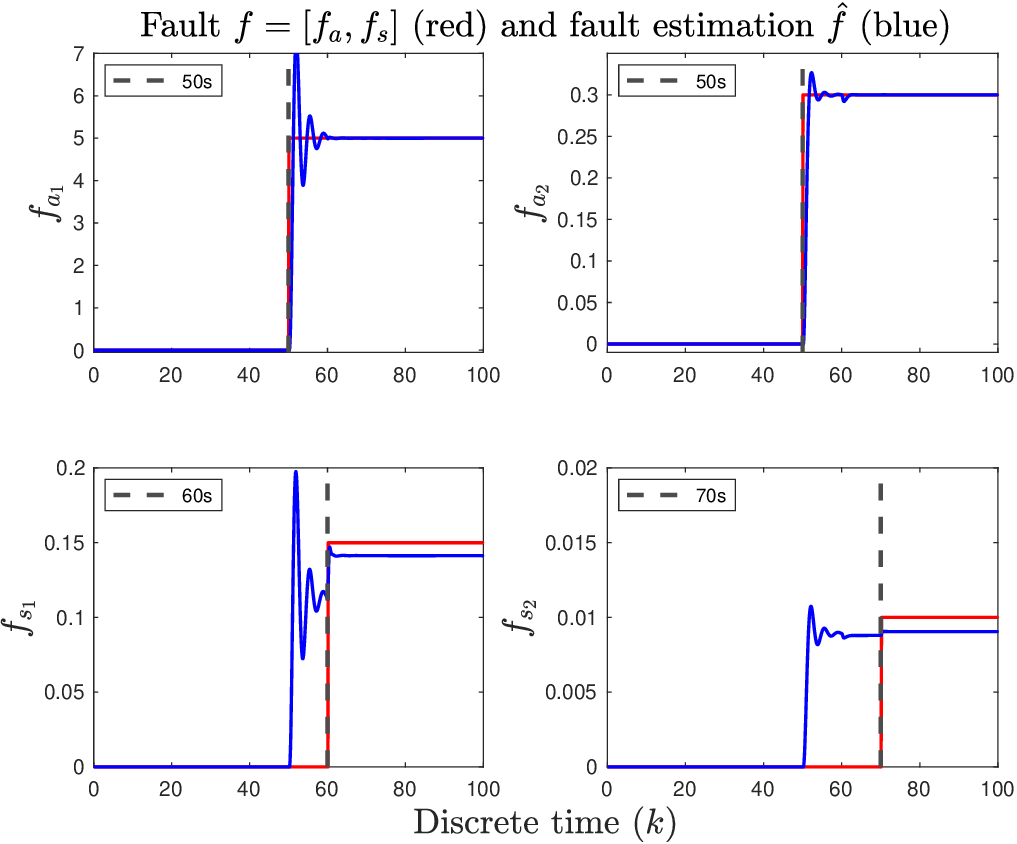}}
    \caption{Numerical simulation results. (a) Closed-loop response under different actuator loss-of-effectiveness factors. (b) Tracking performance comparison between the conventional PID controller and the proposed DDPG controller, where the learning agent is trained using only the yaw channel. (c)--(d) Fault estimation performance of the proposed proportional--integral observer for the considered actuator and actuator--sensor fault scenarios.}
    \label{f4}
\end{figure}
\begin{figure*}[!t]
    \centering
    \subfloat[\label{f5a}\scriptsize FTC under Type-1 actuator faults.]{
    \includegraphics[width=.33\linewidth]{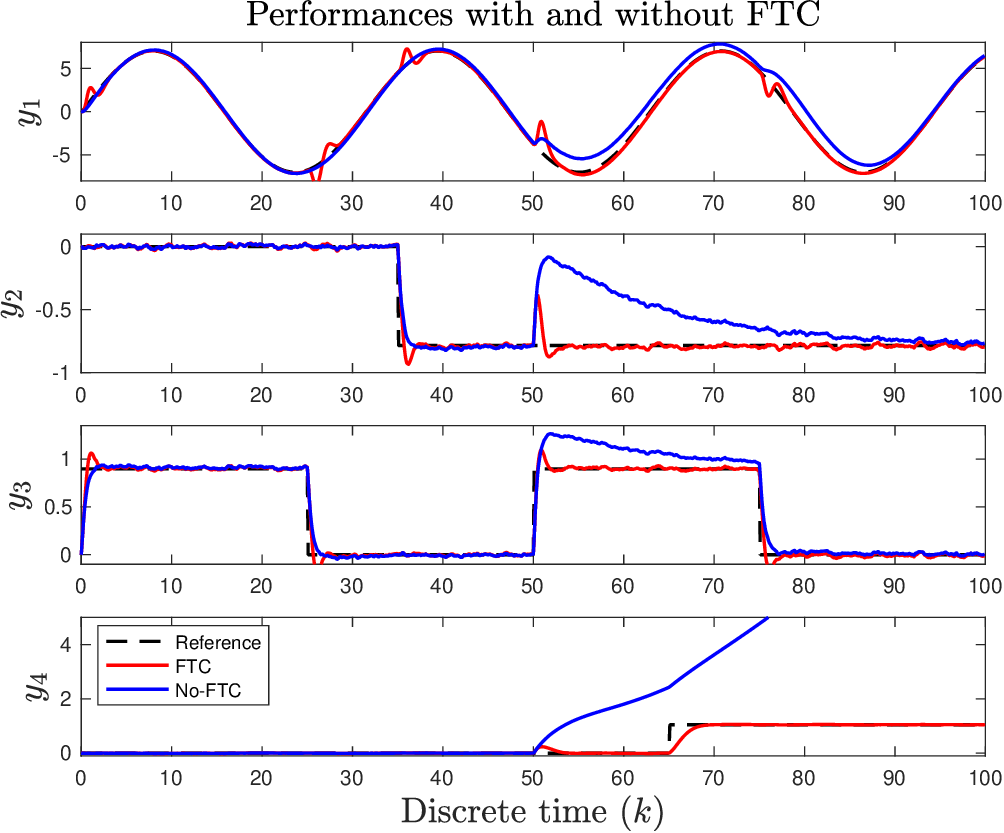}}
    \subfloat[\label{f5b}\scriptsize FTC under Type-2 actuator faults.]{
    \includegraphics[width=.33\linewidth]{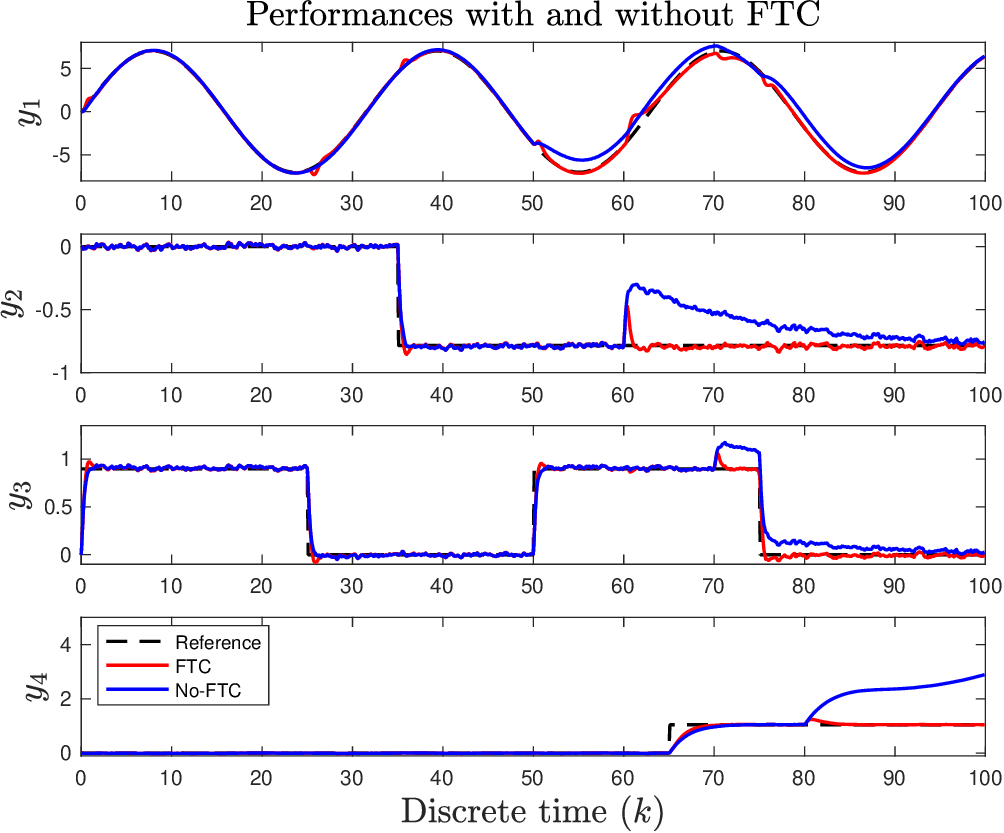}}
    \subfloat[\label{f5c}\scriptsize FTC under Type-3 actuator--sensor faults.]{
    \includegraphics[width=.33\linewidth]{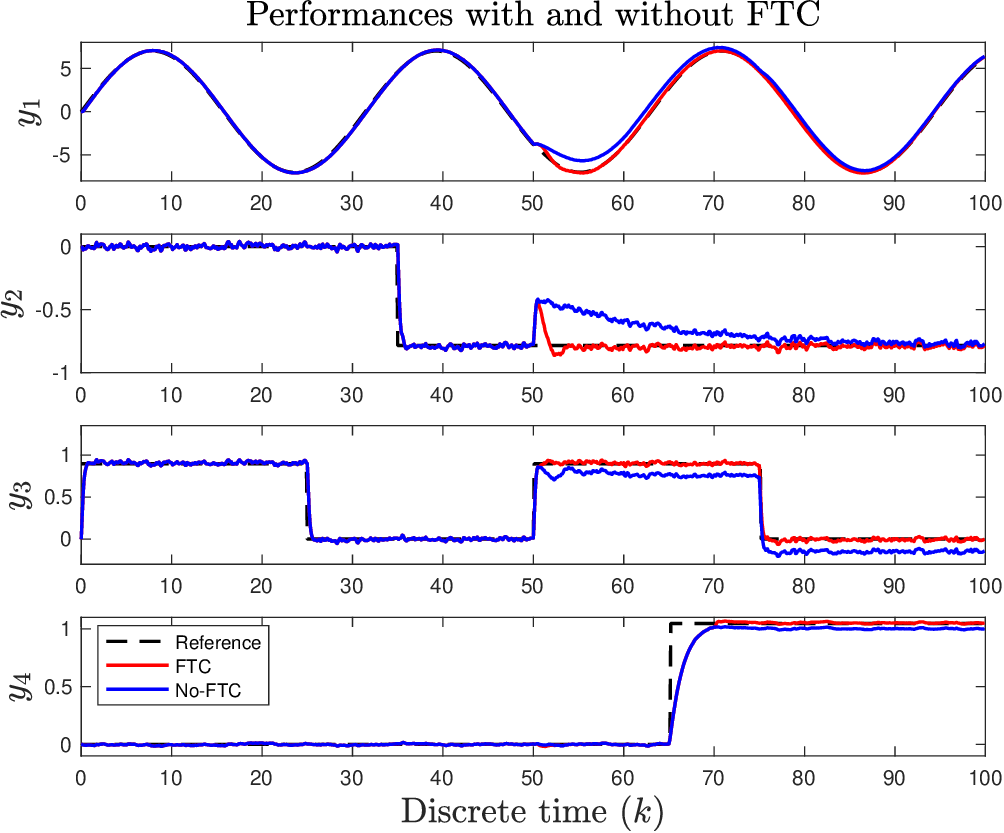}}
    \caption{Performance comparison between the proposed fault-tolerant control (FTC) scheme and the uncompensated controller under the three fault scenarios described in the numerical studies. The corresponding fault estimation results are presented in Figs.~\ref{f4c} and \ref{f4d}.}
    \label{f5}
\end{figure*}
\begin{figure*}[!t]
    \centering
    \subfloat[\label{fig:tv_fault_estimation}\scriptsize Time-varying actuator fault estimation.]{
    \includegraphics[width=.33\linewidth]{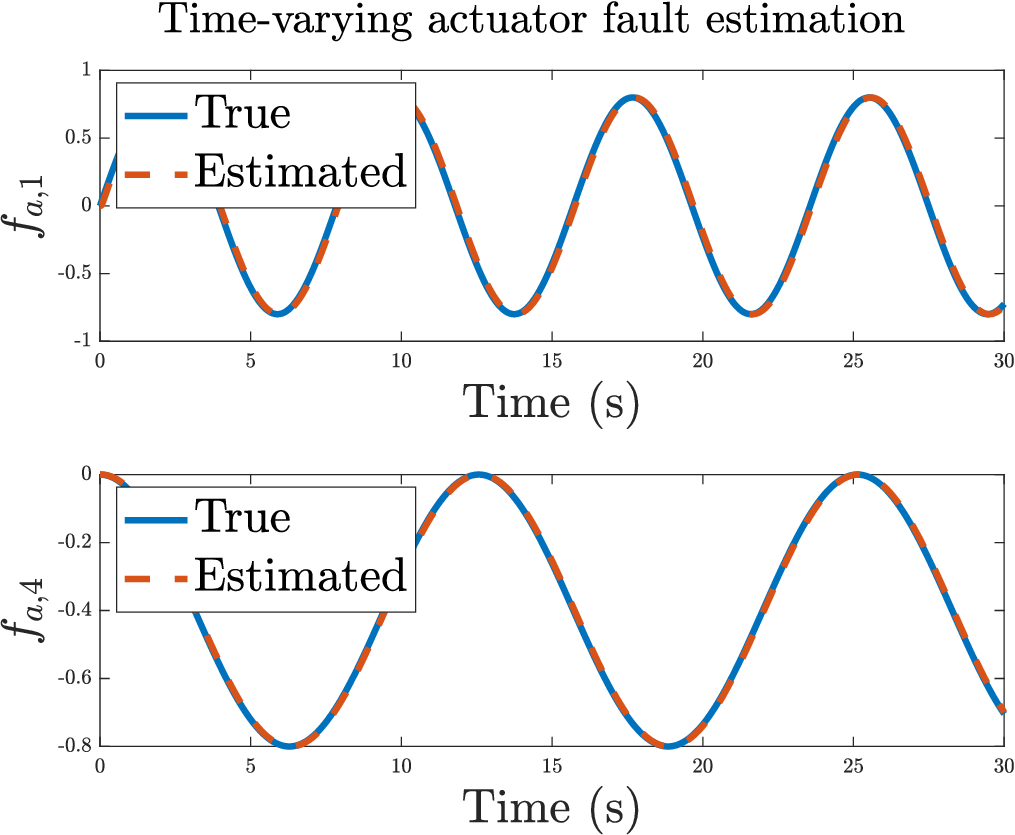}}
    \qquad
    \subfloat[\label{fig:tv_fault_error}\scriptsize Estimation error norm.]{
    \includegraphics[width=.325\linewidth]{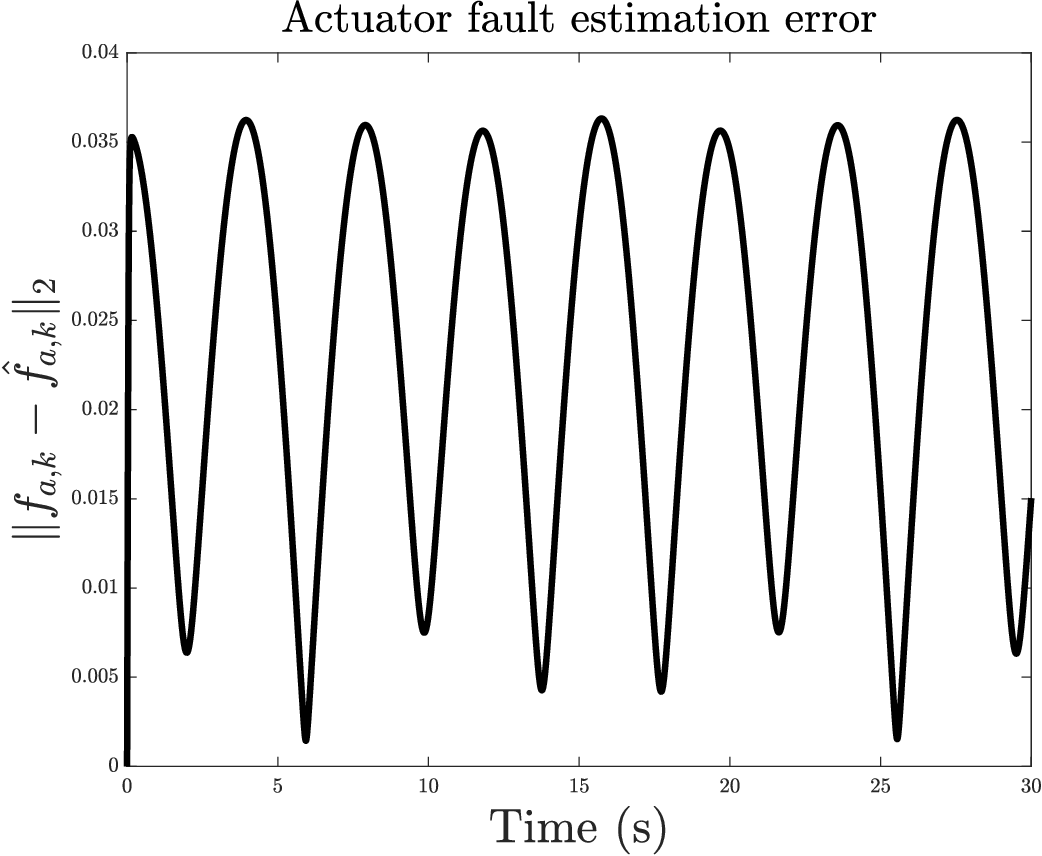}}
    \caption{Performance of the proposed proportional--integral observer for estimating slowly time-varying actuator faults under bounded disturbances.}
    \label{fig:time_varying_actuator_faults}
\end{figure*}

\section{Conclusions}

This paper has presented a fault-tolerant control framework for unmanned aerial vehicles by combining proportional--integral fault estimation with reinforcement learning-based control. A mathematical model of the quadcopter was first developed using both the Newton--Euler and Euler--Lagrange formulations. A proportional--integral observer was then proposed to estimate actuator and sensor faults. The observer analysis established exponential estimation-error convergence for constant faults in the absence of disturbances, and ultimate boundedness for slowly time-varying faults under bounded disturbances. The estimated faults were subsequently incorporated into an active fault-tolerant control architecture and integrated with a Deep Deterministic Policy Gradient controller to improve tracking performance under faulty operating conditions. Numerical simulations illustrated the estimation performance of the PI observer under constant and slowly time-varying faults, including bounded disturbances, and the tracking performance of the integrated fault-tolerant controller under the considered actuator and sensor fault scenarios.

Future work will extend the proposed framework to cooperative multi-agent UAV systems, where distributed fault estimation and fault-tolerant learning-based control can be developed for coordinated missions. Other promising directions include adaptive observer design for rapidly time-varying faults, communication-aware resilient control over networked systems, and reinforcement learning-based resilient control under communication uncertainties and adversarial disturbances \cite{F1,F2,Wafi-MRAC,F3,Wafi-ASCC26-FTC,F4,F5,F6}.

\newpage

\bibliographystyle{ieeetr}
\bibliography{EN-Bibliography}

@ARTICLE{R1,
  author={Shakhatreh, Hazim and Sawalmeh, Ahmad H. and Al-Fuqaha, Ala and Dou, Zuochao and Almaita, Eyad and Khalil, Issa and Othman, Noor Shamsiah and Khreishah, Abdallah and Guizani, Mohsen},
  journal={IEEE Access}, 
  title={Unmanned Aerial Vehicles ({UAV}s): A Survey on Civil Applications and Key Research Challenges}, 
  year={2019},
  volume={7},
  number={},
  pages={48572-48634},
  doi={10.1109/ACCESS.2019.2909530}}

@ARTICLE{R2,
  author={Labib, Nader S. and Brust, Matthias R. and Danoy, Grégoire and Bouvry, Pascal},
  journal={IEEE Access}, 
  title={The Rise of Drones in Internet of Things: A Survey on the Evolution, Prospects and Challenges of Unmanned Aerial Vehicles}, 
  year={2021},
  volume={9},
  number={},
  pages={115466-115487},
  doi={10.1109/ACCESS.2021.3104963}}

@ARTICLE{R3,
  author={Chung, Wonmo and Son, Hungsun},
  journal={IEEE/ASME Transactions on Mechatronics}, 
  title={Fault-Tolerant Control of Multirotor {UAV}s by Control Variable Elimination}, 
  year={2020},
  volume={25},
  number={5},
  pages={2513-2522},
  doi={10.1109/TMECH.2020.2982436}}

@ARTICLE{R5,
  author={Zuo, Zongyu and Liu, Cunjia and Han, Qing-Long and Song, Jiawei},
  journal={IEEE/CAA Journal of Automatica Sinica}, 
  title={Unmanned Aerial Vehicles: Control Methods and Future Challenges}, 
  year={2022},
  volume={9},
  number={4},
  pages={601-614},
  doi={10.1109/JAS.2022.105410}}

@ARTICLE{R6,
  author={Huang, Hailong and Savkin, Andrey V.},
  journal={IEEE Transactions on Vehicular Technology}, 
  title={Navigating {UAV}s for Optimal Monitoring of Groups of Moving Pedestrians or Vehicles}, 
  year={2021},
  volume={70},
  number={4},
  pages={3891-3896},
  doi={10.1109/TVT.2021.3065102}}

@ARTICLE{R7,
  author={Bithas, Petros S. and Nikolaidis, Viktor and Kanatas, Athanasios G. and Karagiannidis, George K.},
  journal={IEEE Transactions on Communications}, 
  title={{UAV}-to-Ground Communications: Channel Modeling and {UAV} Selection}, 
  year={2020},
  volume={68},
  number={8},
  pages={5135-5144},
  doi={10.1109/TCOMM.2020.2992040}}

@ARTICLE{R8,
  author={Wang, Dan and Huang, Jianhao and He, Meng and Huang, Chuan},
  journal={IEEE Wireless Communications Letters}, 
  title={Spectrum Transaction Games for {UAV} Assisted Communications}, 
  year={2022},
  volume={11},
  number={6},
  pages={1216-1219},
  doi={10.1109/LWC.2022.3161387}}

@ARTICLE{R9,
  author={Raja, Gunasekaran and Anbalagan, Sudha and Ganapathisubramaniyan, Aishwarya and Selvakumar, Madhumitha Sri and Bashir, Ali Kashif and Mumtaz, Shahid},
  journal={IEEE Transactions on Vehicular Technology}, 
  title={Efficient and Secured Swarm Pattern Multi-{UAV} Communication}, 
  year={2021},
  volume={70},
  number={7},
  pages={7050-7058},
  doi={10.1109/TVT.2021.3082308}}

@article{R23,
	author = {Katherin Indriawati and Trihastuti Agustinah and Achmad Jazidie},
	title = {Robust Observer-Based Fault Tolerant Tracking Control for Linear Systems with Simultaneous Actuator and Sensor Faults: Application to a DC Motor System},
	journal = {International Review on Modelling and Simulations (IREMOS)},
	volume = {8},
	number = {4},
	year = {2015},
	issn = {2533-1701}
}

@article{R29,
  title={Fault-Tolerant Control Design in Scrubber Plant with Fault on Sensor Sensitivity},
  author={Wafi, Moh Kamalul and Indriawati, Katherin},
  journal={The Journal of Scientific and Engineering Research},
  year={2022},
  volume={9},
  number={2},
  pages={96-104}
}

@book{R33,
  title     = "Flight stability and automatic control 2nd ed.",
  author    = "Nelson, Robert C.",
  year      = 1998,
  publisher = "WCB/McGraw Hill",
  address   = "Boston, Mass."
}

@book{R34,
  title     = "Introduction to Aircraft Flight Mechanics: Performance, Static Stability, Dynamic Stability, and Classical Feedback Control, Second Edition",
  author    = "Yechout, Thomas R. and Morris, Steven L. and Bossert, David E. and Hallgren, Wayne F. and Hall, James K.",
  year      = 2014,
  publisher = "American Institute of Aeronautics and Astronautics (AIAA), Inc",
  address   = "1801 Alexander Bell Drive, Reston, VA 20191-4344"
}

@book{R42,
  title     = "State Estimation and Control for Low-cost Unmanned Aerial Vehicles",
  author    = "Hajiyev, Chingiz and Soken, Halil Ersin and Vural, Sıtkı Yenal",
  year      = 2015,
  publisher = "Springer Chamc",
  address   = "1801 Alexander Bell Drive, Reston, VA 20191-4344",
  doi={10.1007/978-3-319-16417-5}
}

@article{R44,
author = {Niemann, H. H. and Stoustrup, J. and Shafai, B. and Beale, S.},
title = {LTR design of proportional-integral observers},
journal = {International Journal of Robust and Nonlinear Control},
volume = {5},
number = {7},
pages = {671-693},
doi = {https://doi.org/10.1002/rnc.4590050706},
url = {https://onlinelibrary.wiley.com/doi/abs/10.1002/rnc.4590050706},
eprint = {https://onlinelibrary.wiley.com/doi/pdf/10.1002/rnc.4590050706},
year = {1995}
}

@INPROCEEDINGS{R45,
  author={Shafai, B. and Pi, C.T. and Nork, S.},
  booktitle={Proceedings of the 5th Biannual World Automation Congress}, 
  title={Robust fault detection using proportional integral observers}, 
  year={2002},
  volume={14},
  number={},
  pages={515-520},
  doi={10.1109/WAC.2002.1049489}}

@INPROCEEDINGS{R46,
  author={Shafai, Bahram and Moradmand, Anahita},
  booktitle={2020 IEEE Conference on Control Technology and Applications (CCTA)}, 
  title={Design of an Integrated Observer Structure for Robust Fault Detection}, 
  year={2020},
  volume={},
  number={},
  pages={248-253},
  doi={10.1109/CCTA41146.2020.9206387}}

@INPROCEEDINGS{R47,
  author={Vankadari, Madhu Babu and Das, Kaushik and Shinde, Chinmay and Kumar, Swagat},
  booktitle={2018 International Conference on Unmanned Aircraft Systems (ICUAS)}, 
  title={A Reinforcement Learning Approach for Autonomous Control and Landing of a Quadrotor}, 
  year={2018},
  volume={},
  number={},
  pages={676-683},
  doi={10.1109/ICUAS.2018.8453468}}

@ARTICLE{R48,
  author={Oubbati, Omar Sami and Lakas, Abderrahmane and Guizani, Mohsen},
  journal={IEEE Internet of Things Journal}, 
  title={Multiagent Deep Reinforcement Learning for Wireless-Powered UAV Networks}, 
  year={2022},
  volume={9},
  number={17},
  pages={16044-16059},
  doi={10.1109/JIOT.2022.3150616}}

@ARTICLE{R49,
  author={Zhang, Dapeng and Gao, Zhiwei},
  journal={IEEE Access}, 
  title={Fault Tolerant Control Using Reinforcement Learning and Particle Swarm Optimization}, 
  year={2020},
  volume={8},
  number={},
  pages={168802-168811},
  doi={10.1109/ACCESS.2020.3022893}}

@ARTICLE{R50,
  author={Hua, Changsheng and Li, Linlin and Ding, Steven X.},
  journal={IEEE Transactions on Automatic Control}, 
  title={Reinforcement Learning-Aided Performance-Driven Fault-Tolerant Control of Feedback Control Systems}, 
  year={2022},
  volume={67},
  number={6},
  pages={3013-3020},
  doi={10.1109/TAC.2021.3088397}}

@ARTICLE{R51,
  author={Deng, Huifan and Zhao, Youqun and Nguyen, Anh-Tu and Huang, Chao},
  journal={IEEE/ASME Transactions on Mechatronics}, 
  title={Fault-Tolerant Predictive Control With Deep-Reinforcement-Learning-Based Torque Distribution for Four In-Wheel Motor Drive Electric Vehicles}, 
  year={2023},
  volume={28},
  number={2},
  pages={668-680},
  doi={10.1109/TMECH.2022.3233705}}

@ARTICLE{F1,
  author={Liang, Dingguo and He, Zhichen and Ding, Shuyu and Yang, Ying},
  journal={IEEE Transactions on Circuits and Systems I: Regular Papers}, 
  title={Distributed Fault Estimation and Fault-Tolerant Control of Interconnected Systems With Plug-and-Play Features}, 
  year={2024},
  volume={71},
  number={1},
  pages={431-442},
  doi={10.1109/TCSI.2023.3330839}}

@INPROCEEDINGS{F2,
  author={Bhan, Luke and Quinones-Grueiro, Marcos and Biswas, Gautam},
  booktitle={2021 5th International Conference on Control and Fault-Tolerant Systems (SysTol)}, 
  title={Fault Tolerant Control combining Reinforcement Learning and Model-based Control}, 
  year={2021},
  volume={},
  number={},
  pages={31-36},
  doi={10.1109/SysTol52990.2021.9595275}}

@INPROCEEDINGS{F3,
  author={Taslima, Eram and Kamal, Shyam and Saket, R. K. and Dinh, Thach Ngoc},
  booktitle={2025 American Control Conference (ACC)}, 
  title={An Adaptive Observer Design for State and Parameter Estimation of Quantum Systems via Averaging Theory}, 
  year={2025},
  volume={},
  number={},
  pages={4885-4890},
  doi={10.23919/ACC63710.2025.11107918}}

@ARTICLE{F4,
  author={Nguyen, Truong-Son and Cao, Huy-Nhat and Thai, Nguyen Huu and Pham, Minh Trien},
  journal={IEEE Transactions on Automation Science and Engineering}, 
  title={C-CBF: Communication-Aware Control Barrier Functions for Resilient Multi-Robot Connectivity}, 
  year={2026},
  volume={23},
  number={},
  pages={5596-5614},
  doi={10.1109/TASE.2026.3667964}}

@ARTICLE{F5,
  author={Sahu, Abhijeet and Venkatraman, Venkatesh and Macwan, Richard},
  journal={IEEE Access}, 
  title={Reinforcement Learning Environment for Cyber-Resilient Power Distribution System}, 
  year={2023},
  volume={11},
  number={},
  pages={127216-127228},
  doi={10.1109/ACCESS.2023.3282182}}

@INPROCEEDINGS{F6,
  author={Gao, Bing and Bu, Bing},
  booktitle={2021 IEEE International Intelligent Transportation Systems Conference (ITSC)}, 
  title={Deep Reinforcement Learning-based Resilient Control Method for CBTC Systems through Train-to-Train Communications under Adversarial Attacks}, 
  year={2021},
  volume={},
  number={},
  pages={3679-3684},
  doi={10.1109/ITSC48978.2021.9564835}}

@INPROCEEDINGS{Wafi-ASCC26-DNC,
  author={Wafi, Moh Kamalul and Ataka, Ahmad and Nazaruddin, Yul Y. and Jayawardhana, Bayu},
  booktitle={2026 15th Asian Control Conference (ASCC)}, 
  title={Distributed Nonlinear Control of Networked Two-Wheeled Robots Under Adversarial Interactions}, 
  year={2026},
  volume={},
  number={},
  pages={1524-1529},
  doi={10.1109/ASCC70886.2026.11701762}
}

@INPROCEEDINGS{Wafi-ASCC26-FTC,
  author={Wafi, Moh Kamalul and Nugraha, Yurid E. and Widjiantoro, Bambang L. and Indriawati, Katherin},
  booktitle={2026 15th Asian Control Conference (ASCC)}, 
  title={Cooperative $\mathcal{H}_{\infty}$ Fault-Tolerant Tracking with ISS Guarantees for Networked Systems with Sensor Faults}, 
  year={2026},
  volume={},
  number={},
  pages={1468-1473},
  doi={10.1109/ASCC70886.2026.11702178}
}

@article{Wafi-MRAC,
   title={Model reference adaptive control of networked systems with state and input delays},
   volume={14},
   ISSN={2088-8708},
   DOI={10.11591/ijece.v14i5.pp5055-5063},
   number={5},
   journal={International Journal of Electrical and Computer Engineering (IJECE)},
   publisher={Institute of Advanced Engineering and Science},
   author={Wafi, Moh Kamalul and Indriawati, Katherin and Widjiantoro, Bambang L.},
   year={2024},
   month=oct, pages={5055} 
}

@INPROCEEDINGS{Wafi-CoDIT,
  author={Wafi, Moh Kamalul and Hajian, Rozhin and Shafai, Bahram and Siami, Milad},
  booktitle={2023 9th International Conference on Control, Decision and Information Technologies (CoDIT)}, 
  title={Advancing Fault-Tolerant Learning-Oriented Control for Unmanned Aerial Systems}, 
  year={2023},
  volume={},
  number={},
  pages={1688-1693},
  doi={10.1109/CoDIT58514.2023.10284443}
}

@article{Wafi-SIAM-arXiv,
    title={On the Comparison of Reinforcement Learning and Adaptive Control for Linear Systems under Packet Loss and Uncertainty}, 
    author={Moh Kamalul Wafi},
    journal = {arXiv preprint arXiv:2606.32003},
    year = {2026},
}

@INPROCEEDINGS{Wafi-SIAM,
author = {Moh. Kamalul Wafi and Milad Siami},
title = {A Comparative Analysis of Reinforcement Learning and Adaptive Control Techniques for Linear Uncertain Systems},
booktitle = {2023 Proceedings of the Conference on Control and its Applications (CT)},
chapter = {},
year={2023},
pages = {25-32},
doi = {10.1137/1.9781611977745.4},
eprint = {https://epubs.siam.org/doi/pdf/10.1137/1.9781611977745.4}
}

@article{WAFI-JRC, 
    title={Non-Linear Estimation using the Weighted Average Consensus-Based Unscented Filtering for Various Vehicles Dynamics towards Autonomous Sensorless Design}, 
    volume={4}, 
    number={1}, 
    journal={Journal of Robotics and Control (JRC)}, 
    author={Widjiantoro, Bambang L. and Wafi, Moh Kamalul and Indriawati, Katherin}, 
    year={2023}, 
    month={Mar.}, 
    pages={95–107},
    url={https://journal.umy.ac.id/index.php/jrc/article/view/16164}, 
    DOI={10.18196/jrc.v4i1.16164},
}

@article{Wafi-Hydraulic-arXiv,
    title={Estimation and Fault Detection on Hydraulic System with Adaptive-Scaling Kalman and Consensus Filtering}, 
    author={Moh Kamalul Wafi},
    journal = {arXiv preprint arXiv:2304.04122},
    year = {2023},
}

@article{Wafi-AIP-arXiv,
    title={Neural-Enhanced Micro-Kalman Filtering for Satellite Tracking: A Comparative Study}, 
    author={Moh Kamalul Wafi},
    journal = {arXiv preprint arXiv:2607.06830},
    year = {2026},
}

@ARTICLE{Wafi-AIP,
    author = {Wafi, Moh Kamalul},
    title = {Filtering module on satellite tracking},
    journal = {AIP Conference Proceedings},
    volume = {2088},
    number = {1},
    pages = {020045},
    year = {2019},
    doi = {10.1063/1.5095297},
    URL = {https://aip.scitation.org/doi/abs/10.1063/1.5095297}
}

\end{document}